\documentclass[journal]{IEEEtran}
\usepackage{amsmath,amsfonts}
\usepackage{algorithmic}
\usepackage{algorithm}
\usepackage{array}
\usepackage[caption=false,font=normalsize,labelfont=sf,textfont=sf]{subfig}
\usepackage{textcomp}
\usepackage{stfloats}
\usepackage{url}
\usepackage{verbatim}
\usepackage{graphicx}
\usepackage{cite}
\usepackage{bm}
\usepackage{amssymb} 
\usepackage{mdwmath}
\usepackage{mdwtab}
\usepackage{eqparbox}
\usepackage{cuted}
\usepackage[hidelinks]{hyperref}

\newtheorem{proposition}{\bfseries Proposition}

\newtheorem{remark}{\bfseries Remark}
\newtheorem{corollary}{\bfseries Corollary}

\begin{document}

\title{Beyond Beamforming: Phase-and-Gain Channel Shaping via Rotatable Antenna Arrays}

\author{Xingxiang Peng, Qingqing Wu, Ziyuan Zheng, Wen Chen, and Guangchi Zhang 
\thanks{
    X. Peng, Q. Wu, Z. Zheng, and W. Chen are with the School of Integrated Circuits, Shanghai Jiao Tong University, 200240, China 
    (e-mail: 
    \href{mailto:peng_xingxiang@sjtu.edu.cn}{\nolinkurl{peng_xingxiang@sjtu.edu.cn}};
    \href{mailto:qingqingwu@sjtu.edu.cn}{\nolinkurl{qingqingwu@sjtu.edu.cn}};
    \href{mailto:zhengziyuan2024@sjtu.edu.cn}{\nolinkurl{zhengziyuan2024@sjtu.edu.cn}};
    \href{mailto:wenchen@sjtu.edu.cn}{\nolinkurl{wenchen@sjtu.edu.cn}}).
    G. Zhang is with the School of Information Engineering, Guangdong University of Technology, 510006, China 
    (e-mail:
    \href{mailto:gczhang@gdut.edu.cn}{\nolinkurl{gczhang@gdut.edu.cn}}).
    }
} 

\IEEEaftertitletext{\vspace{-2\baselineskip}}

\maketitle

\begin{abstract}
    This paper investigates geometry-reconfigurable transmission for multiuser communication systems enabled by a rotatable antenna array. In contrast to conventional fixed arrays, the proposed architecture jointly exploits array pose adjustment and element-level boresight steering, thereby reshaping both the array-induced phase responses and the direction-dependent channel gains. We formulate a weighted sum-rate maximization problem that jointly optimizes the transmit beamformers, array pose, and element boresights under practical visibility and steering constraints. To reveal the underlying design principles, we first provide a geometric interpretation via zero-forcing analysis, showing that the resulting rates stem from both channel-strength enhancement and spatial-separability improvement. Specifically, array-pose rotation improves inter-user channel orthogonality even with isotropic elements, whereas directional elements introduce a tradeoff between phase-based spatial separation and boresight-dependent gain alignment. Motivated by these insights, we develop an efficient optimization framework that jointly coordinates transmit beamforming, array-pose adaptation, and element-boresight steering to exploit the geometry-induced phase-and-gain channel-shaping capability. Simulation results demonstrate that the proposed joint design outperforms fixed-array, pose-only, and boresight-only benchmarks, with larger gains achieved under more directive element patterns and tighter boresight-steering constraints.
\end{abstract}


\section{Introduction}

\IEEEPARstart{F}{uture} sixth-generation (6G) wireless networks are expected to support emerging applications such as immersive extended reality, digital twins, and integrated sensing and communication, which require highly efficient use of spatial, spectral, and energy resources \cite{10555049}. Multi-antenna techniques, including massive multiple-input multiple-output (MIMO), hybrid beamforming, and cell-free architectures, have been widely investigated to improve spectral efficiency in multiuser wireless systems \cite{10379539,10496996,11053128}. However, these techniques mainly optimize signal-domain variables over channels determined by fixed array geometry and propagation environments. Consequently, their performance is still limited by unfavorable channel conditions, such as strong inter-user correlation, limited spatial rank, and directional mismatch.

To move beyond signal adaptation over fixed channels, recent studies have explored physical-domain reconfigurability to actively reshape wireless propagation. Intelligent reflecting surfaces (IRSs) represent a typical environment-side solution, where a large number of nearly passive reflecting elements are jointly tuned to create favorable auxiliary propagation paths, enhance coverage, and suppress interference \cite{8910627,9998527,10795216,11007277}. In parallel with environment-side reconfigurable technologies, movable antennas (MAs) and fluid antennas (FAs) have recently emerged as representative transceiver-side approaches for exploiting spatial-domain reconfigurability \cite{10753482,10906511}. MAs improve wireless links by relocating antenna elements within a prescribed region \cite{10286328,10243545,10354003}, whereas FAs select or activate favorable ports over a reconfigurable aperture \cite{11049889,10794752,11039166}, thereby exploiting spatial channel variations, enlarging the effective aperture, and reducing channel correlation. However, most existing MA/FA investigations adopt isotropic element models, whereas practical base-station (BS) arrays typically employ downtilted directional elements, whose higher gain and spatial focusing capability deserve explicit consideration in transceiver-side physical reconfigurability.

A natural way to exploit directional radiation is to enable antenna-orientation reconfigurability at the transceiver \cite{11222668,11427014,11489290}. In this regard, rotatable antennas (RAs) provide a geometry-domain mechanism for shaping the effective channel by coupling controllable orientations with direction-dependent element patterns. For planar arrays, such orientation reconfigurability can act at two levels. At the array level, rotating the uniform planar array (UPA) changes the whole aperture orientation, thereby modifying both the element positions and the local viewing directions. Hence, it reshapes the array-induced phase progression, enhances inter-user spatial separability, and provides coarse directional coverage \cite{11520277,11534579,11206404}. At the element level, boresight steering adjusts the radiation direction of each antenna element, which refines the direction-dependent gains toward the resulting local user directions \cite{11134688,11520842,arxiv1}. These two levels are inherently coupled rather than independently additive: a pose that is favorable for spatial separability may not provide sufficient directional-gain alignment, whereas boresight steering under a fixed pose cannot modify the aperture-induced phase geometry. Therefore, a joint treatment of array pose and element boresights is needed to understand and exploit the phase-and-gain channel-shaping capability of directional arrays under practical steering constraints.

Motivated by the above discussion, this paper investigates a geometry-reconfigurable transmission architecture for multiuser communication systems enabled by a rotatable UPA with steerable element boresights. Instead of treating array pose and element boresights as separate design knobs, we optimize them as two coupled geometry-domain DoFs together with transmit beamforming. The UPA pose reshapes the array-induced phase geometry and provides coarse directional coverage, while the element boresights refine direction-dependent gains under the resulting local user directions. This joint phase-and-gain channel-shaping perspective leads to a challenging signal--geometry coupled WSR maximization problem and motivates the analysis and algorithm developed in this paper. Overall, the main contributions of this paper are summarized as follows.

\begin{itemize}
    \item We formulate a weighted sum-rate (WSR) maximization problem for rotatable-UPA-enabled multiuser downlink transmission, where the transmit beamformers, UPA pose, and element boresights are jointly optimized. The formulation captures both aperture-induced phase responses and direction-dependent radiation gains, leading to a highly nonconvex signal-geometry coupled design problem.

    \item We first provide a zero-forcing (ZF) interpretation to reveal the physical roles of UPA pose rotation and element boresight steering. By decomposing the ZF rate into channel-strength and spatial-separability terms, we show that UPA rotation improves spatial separability by reshaping array-induced phase differences even with isotropic elements. For directional elements, the ZF interpretation further characterizes a fundamental tradeoff between phase-based spatial separability and boresight-dependent directional-gain alignment, thereby necessitating the joint UPA-pose and element-boresight design.

    \item Motivated by the above geometric insights, we develop an efficient optimization framework for the general multiuser problem. The beamforming block is handled by the weighted minimum mean-square error (WMMSE) method, while the UPA pose is updated using analytical WSR gradients, a trust-region ascent direction, and feasible Armijo search. The element boresights are further refined via tangent-space gradients and a Frank--Wolfe (FW)-type update with a closed-form spherical-cap oracle. These block-wise updates preserve feasibility and yield a nondecreasing WSR sequence.

    \item Simulation results validate both the geometric interpretation and the effectiveness of the proposed algorithm. The joint design consistently outperforms fixed-UPA, pose-only, boresight-only, and isotropic benchmarks, showing that antenna directivity is beneficial only when properly matched to the user geometry and that UPA pose and element boresights enable complementary global adjustment of array-induced phase responses and local refinement of direction-dependent gains.
\end{itemize}

The remainder of this paper is organized as follows. Section II introduces the system model and formulates the WSR maximization problem. Section III provides a ZF-based geometric interpretation of pose--boresight channel shaping. Section IV develops a joint optimization framework for beamforming and array geometry. Section V presents simulation results, and Section VI concludes this paper.

\section{System Model and Problem Formulation}
\label{sysmodel}

\begin{figure}[t]
	\begin{center}
		\includegraphics[width=0.48\textwidth]{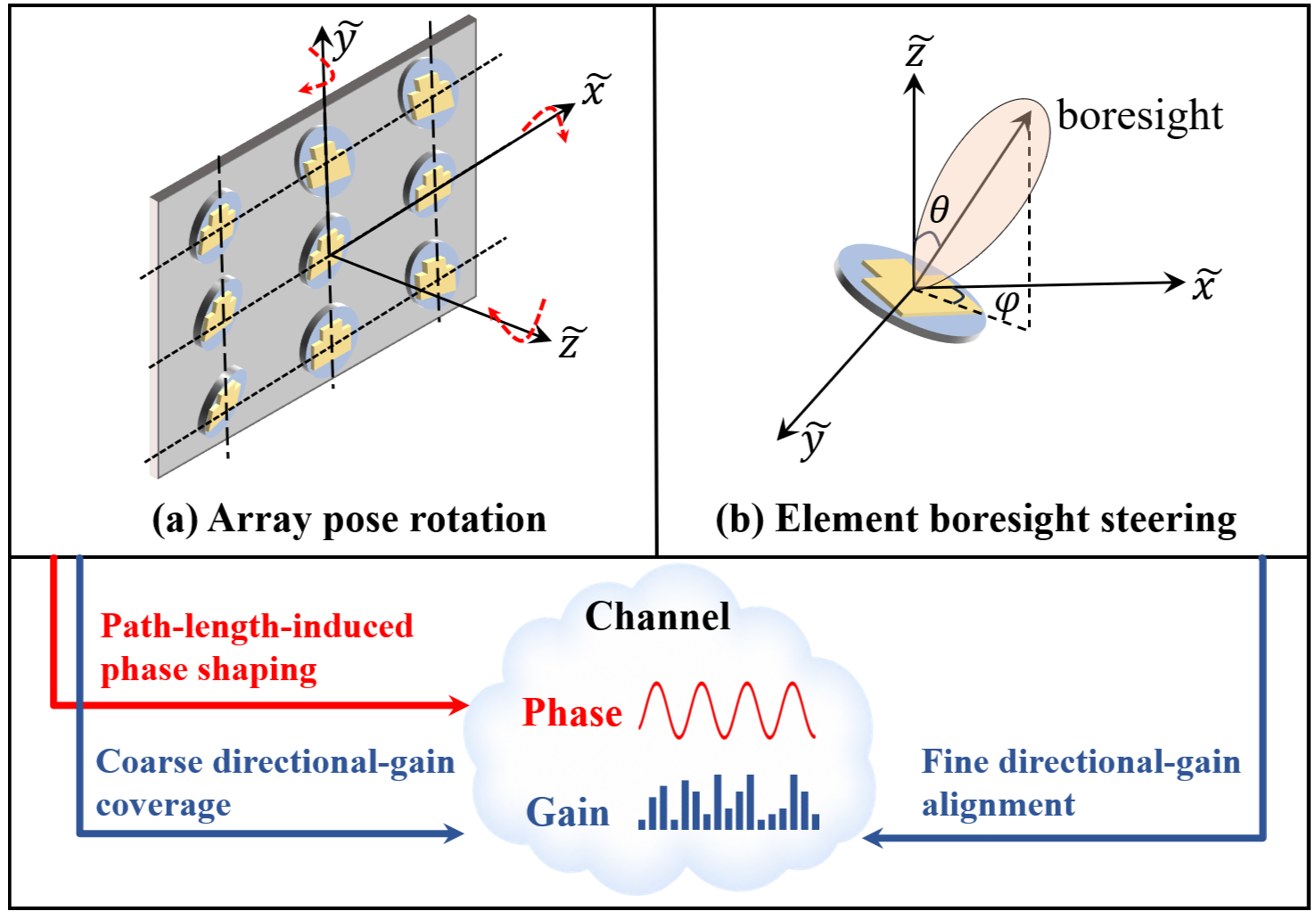}
		\caption{Illustration of the proposed rotatable UPA for phase-and-gain channel shaping via (a) Array pose rotation and (b) Element boresight steering.}
		\label{fig:geometry_dofs}
	\end{center}
\end{figure}

We consider a multiuser downlink system where a BS equipped with a rotatable UPA serves $K$ single-antenna users with isotropic reception. As illustrated in Fig.~\ref{fig:geometry_dofs}, our proposed UPA provides two geometry-domain DoFs: UPA pose rotation and element boresight steering. The positions of the UPA center and user $k$ in the global coordinate system (GCS) are denoted by $\bm p_0\in\mathbb{R}^3$ and $\bm q_k\in\mathbb{R}^3$, respectively. The UPA consists of $M=M_x\times M_y$ elements with inter-element spacing $d$. The element index $m\in\{1,\dots,M\}$ is mapped to the two-dimensional index $(m_x,m_y)$ in row-major order as
\begin{align}
    m = m_x + (m_y-1)M_x,
\end{align}
where $m_x\in\{1,\ldots,M_x\}$ and $m_y\in\{1,\ldots,M_y\}$. The position of element $m$ in the UPA local coordinate system (LCS), with the UPA center as the origin, is given by
\begin{align}
    \tilde{\bm p}_{m}
    =
    \left[
    \left(m_x-\frac{M_x+1}{2}\right)d,\,
    \left(m_y-\frac{M_y+1}{2}\right)d,\,
    0
    \right]^T.
    \label{eq:tx-upa-coord_MU}
\end{align}

We next describe the two geometric DoFs of the rotatable UPA, then present the corresponding channel model and problem formulation.

\subsection{Array Pose Rotation}
The UPA pose is characterized by a rotation matrix $\bm R\in\mathrm{SO}(3)$, where $\mathrm{SO}(3)$ denotes the three-dimensional special orthogonal group. The matrix $\bm R$ maps vectors from the UPA LCS to the GCS, and its columns correspond to the global directions of the UPA local axes. For illustration, one possible Euler-angle parameterization of $\bm R$ is given by
\begin{align}
   \bm R=\bm R_x(\alpha_x)\,\bm R_y(\alpha_y)\,\bm R_z(\alpha_z),
\end{align}
under the active-rotation convention for column vectors, where the rightmost rotation acts first. Here, $\bm R_x(\cdot)$, $\bm R_y(\cdot)$, and $\bm R_z(\cdot)$ denote the standard active rotation matrices about the global $x$-, $y$-, and $z$-axes, respectively. The Euler-angle expression above is used only to illustrate the physical meaning of the rotation matrix. In the sequel, the UPA pose is represented directly by $\bm R\in\mathrm{SO}(3)$. Accordingly, the global position of element $m$ is given by
\begin{align}
    \bm p_m
    =
    \bm p_0+\bm R\tilde{\bm p}_m .
    \label{eq:Tx-L2G_MU}
\end{align}
To ensure that all users remain in the visible front half-space of the UPA, we impose
\begin{align}
    (\bm R\bm e_z)^T(\bm q_k-\bm p_0)\ge0,\quad \forall k,
    \label{eq:UPA_front_constraint}
\end{align}
where $\bm e_z=[0,0,1]^T$ is the local $z$-axis and $\bm R\bm e_z$ is the global normal direction of the UPA plane. This constraint prevents the UPA from pointing away from any served user and is consistent with the front-side radiation model adopted below.

\subsection{Element Boresight Steering}
Each antenna element independently steers its boresight in the UPA LCS. The local boresight of element $m$ is denoted by a unit vector $\tilde{\bm f}_m$ and parameterized by the zenith--azimuth angles $(\theta_m,\phi_m)$ as
\begin{align}
    \tilde{\bm f}_{m}
    =
    \left[
    \sin\theta_m\cos\phi_m,\,
    \sin\theta_m\sin\phi_m,\,
    \cos\theta_m
    \right]^T,
    \label{eq:local-boresight_MU}
\end{align}
where $\theta_m$ is measured from the local $z$-axis and $\phi_m$ is measured in the local $x$-$y$ plane. To account for practical steering limitations, the local boresight is constrained within the following spherical-cap feasible set:
\begin{align}
    \tilde{\bm f}_m\in\mathcal F
    \triangleq
    \left\{
    \tilde{\bm f}\in\mathbb R^3
    ~\middle|~
    \|\tilde{\bm f}\|_2=1,\,
    \tilde{\bm f}^T\bm e_z\ge\cos\theta_{\max}
    \right\}.
    \label{eq:ori_feasible_set}
\end{align}
The unit-norm constraint ensures that $\tilde{\bm f}_m$ represents a valid boresight direction, while the spherical-cap constraint restricts its maximum deviation from the local array normal $\bm e_z$ to $\theta_{\max}$. Given the UPA pose $\bm R$, the corresponding global boresight direction of element $m$ is
\begin{align}
    \bm f_m
    =
    \bm R\tilde{\bm f}_m .
    \label{eq:tx-global-boresight_MU}
\end{align}
This global boresight direction determines the direction-dependent radiation gain in the channel model.

\subsection{Channel Model}
The directional power gain of each antenna element is modeled by a front-side cosine pattern. For a unit boresight direction $\bm f$ and a unit departure direction $\bm d$, the power pattern is given by \cite{Ant2016,11489290}
\begin{align}
    G(\bm f,\bm d)
    =
    \kappa_{\max}
    \big[\bm f^T\bm d\big]_+^{2p},
    \label{eq:RA_pattern_MU}
\end{align}
where $[x]_+\triangleq\max\{x,0\}$, $\bm f^T\bm d$ is the cosine of the angular offset between the boresight and departure directions, and $\kappa_{\max}=2(2p+1)$ is chosen for power normalization. The positive-part operation restricts radiation to the visible front half-space, while $p\ge0$ controls the main-lobe width, with a larger $p$ corresponding to a narrower pattern. In particular, for $p=0$, we adopt the convention that the gain equals $\kappa_{\max}=2$ over the visible front half-space and is zero over the back half-space, corresponding to a front-side isotropic pattern.

For the direct link from transmit element $m$ to user $k$, let $r_{m,k}=\|\bm q_k-\bm p_m\|$ denote the propagation distance. The corresponding unit departure direction from element $m$ to user $k$ is given by
\begin{align}
    \bm d_{m,k}
    =
    \frac{\bm q_k-\bm p_m}{r_{m,k}},
    \label{eq:los_direction}
\end{align}
which determines the directional gain of element $m$ toward user $k$. The line-of-sight (LoS) channel coefficient is
\begin{align}
    h_{k,m}^{\mathrm{LoS}}
    =
    \frac{\sqrt{\beta_0 G(\bm f_m,\bm d_{m,k})}}{r_{m,k}}
    \exp\!\left(
    -j\frac{2\pi}{\lambda}r_{m,k}
    \right),
    \label{eq:h_mk_los}
\end{align}
where $\lambda$ denotes the carrier wavelength and $\beta_0=(4\pi/\lambda)^{-2}$ denotes the free-space path-power gain at the reference distance of $1$ m according to the Friis transmission formula.

We further consider $Q$ scattering clusters with locations $\{\bm s_q\}_{q=1}^{Q}$ and effective radar cross-sections $\{\zeta_q\}_{q=1}^{Q}$ to model the non-line-of-sight (NLoS) propagation. For the BS-to-cluster segment, let $r_{m,q}=\|\bm s_q-\bm p_m\|$ denote the propagation distance. The unit departure direction from element $m$ to cluster $q$ is expressed as
\begin{align}
    \bm d_{m,q} =
    \frac{\bm s_q-\bm p_m}{r_{m,q}},
    \label{eq:nlos_bs_cluster_direction}
\end{align}
which determines the directional gain of element $m$ toward the scatterer $q$. For the cluster-to-user segment, let $r_{q,k}=\|\bm q_k-\bm s_q\|$ denote the propagation distance. The NLoS component from element $m$ to user $k$ via cluster $q$ is modeled as
\begin{align}
    h_{k,m}^{(q)}
    &=
    \frac{\sqrt{\beta_0 G(\bm f_m,\bm d_{m,q})}}{r_{m,q}} \sqrt{\frac{\zeta_q}{4\pi r_{q,k}^{2}}} \notag
    \\
    &\qquad\,\,
    \exp\!\left(
    -j\frac{2\pi}{\lambda}(r_{m,q}+r_{q,k})
    +j\chi_{q,k}
    \right),
    \label{eq:h_mk_nlos_q}
\end{align}
where $\sqrt{\zeta_q/(4\pi r_{q,k}^{2})}$ accounts for the scattering attenuation from cluster $q$ to user $k$, and $\chi_{q,k}\sim\mathcal U[0,2\pi)$ denotes the random scattering phase. Summing over all clusters yields
\begin{align}
    h_{k,m}^{\mathrm{NLoS}}
    =
    \sum_{q=1}^{Q} h_{k,m}^{(q)} .
    \label{eq:h_mk_nlos}
\end{align}

Combining both LoS and NLoS components, the channel coefficient from element $m$ to user $k$ is given by
\begin{align}
    h_{k,m}
    =
    h_{k,m}^{\mathrm{LoS}}
    +
    h_{k,m}^{\mathrm{NLoS}} .
    \label{eq:h_mk_overall}
\end{align}
To match the input-output convention used below, the BS-to-user-$k$ channel vector is defined as
\begin{align}
    \bm h_k
    =
    [h_{k,1},h_{k,2},\ldots,h_{k,M}]^H .
    \label{eq:user_channel_vector}
\end{align}

\begin{remark}
    \label{rem:channel_geometric_roles}
    The channel model highlights the complementary roles of the two geometric DoFs. The UPA pose $\bm R$ changes the element positions and maps the local boresights to global directions, thereby shaping both propagation phases and directional channel gains. For a fixed UPA pose, the element boresights adjust the direction-dependent gains without altering the element positions or propagation phases. Hence, the two geometric DoFs enable coupled phase-and-gain channel shaping. When $p=0$, the gain becomes angle-independent for front-side active paths; thus, boresight steering provides no continuous gain-shaping freedom for such paths, whereas UPA pose rotation can still reshape the relative phases through the pose-dependent element positions.
\end{remark}

\subsection{Problem Formulation}
For user $k$, let $\bm w_k\in\mathbb C^{M}$ denote its beamforming vector and $s_k\sim\mathcal{CN}(0,1)$ denote its information symbol. The received signal at user $k$ is given by
\begin{align}
    y_k
    =
    \bm h_k^H\bm w_k s_k
    +
    \sum_{i\ne k}\bm h_k^H\bm w_i s_i
    +
    n_k,
    \label{eq:rx_signal}
\end{align}
where $n_k\sim\mathcal{CN}(0,\sigma^2)$ denotes the additive white Gaussian noise. The resulting achievable rate for each user is given by
\begin{align}
    r_k
    =
    \log_2\!\left(
    1+
    \frac{|\bm h_k^{H}\bm w_k|^2}
    {\sum_{i\ne k}|\bm h_k^{H}\bm w_i|^2+\sigma^2}
    \right).
    \label{eq:rate_k}
\end{align}
Let $\{\omega_k\ge 0\}$ denote the user weights. We jointly optimize the transmit beamformers, the UPA pose, and the element boresights to maximize the system WSR:
\begin{align}
    \max_{\{\bm w_k\},\,\bm R,\,\{\tilde{\bm f}_m\}}
    \quad
    & \sum_{k=1}^K \omega_k r_k
    \label{prob:WSR_MU}\\
    \mathrm{s.t.}\qquad\,\,
    & \sum_{k=1}^K\|\bm w_k\|_2^2\le P_{\max},
    \label{const:power_MU}\\
    & \bm R\in\mathrm{SO}(3),
    \label{const:SO3_MU}\\
    & (\bm R\bm e_z)^T(\bm q_k-\bm p_0)\ge0,\quad \forall k,
    \label{const:front_MU}\\
    & \tilde{\bm f}_m\in\mathcal F,\quad \forall m.
    \label{const:cap_MU}
\end{align}
Here, \eqref{const:power_MU} ensures that the transmit power does not exceed the maximum budget $P_{\max}$, \eqref{const:SO3_MU} restricts the UPA pose to a valid three-dimensional rotation, \eqref{const:front_MU} ensures that all users lie in the visible front half-space of the UPA, and \eqref{const:cap_MU} restricts the feasible steering range of each element boresight.

\section{Geometric Insights into Phase-and-Gain Channel Shaping}
\label{sec:insight_pose_orientation}

Before developing the general multiuser solution, this section uses a canonical far-field LoS setting to reveal how UPA pose rotation and element-boresight steering shape the effective channels under ZF transmission. We first decompose the ZF rate into channel-strength and spatial-separability terms. Based on this decomposition, we show that for isotropic elements, the boresight DoF degenerates for direction-dependent gain shaping, while UPA pose rotation remains effective by reshaping array-induced phase differences to reduce inter-user channel correlation. In contrast, for directional elements, the UPA pose and element boresights become coupled: pose rotation shapes array-phase-based spatial separability and provides coarse directional-gain alignment, whereas boresight steering further refines element-level directional-gain alignment. This coupling leads to a geometric tradeoff that motivates joint pose-and-boresight design.

\subsection{Rate Decomposition: Channel Strength and Separability}
\label{subsec:two_user_setting}

Consider a two-user far-field LoS setting. Let $\ell_k=\|\bm q_k-\bm p_0\|_2$ and $\bm d_k=(\bm q_k-\bm p_0)/\ell_k$ denote the distance and the departure direction from the UPA center to user $k$, respectively. For a given UPA pose $\bm R$, the corresponding local direction is $\tilde{\bm d}_k=\bm R^T\bm d_k$. Under the far-field approximation with respect to the array aperture, we have $r_{m,k}\approx \ell_k-\tilde{\bm p}_m^T\tilde{\bm d}_k$ and $\bm d_{m,k}\approx\bm d_k$. Hence, by defining $G_{m,k}\triangleq G(\bm R\tilde{\bm f}_m,\bm d_k)=\kappa_{\max}[\tilde{\bm f}_m^T\tilde{\bm d}_k]_+^{2p}$, the LoS channel coefficient in \eqref{eq:h_mk_los} can be approximated as
\begin{align}
    h_{k,m}
    \approx
    \frac{\sqrt{\beta_0 G_{m,k}}}{\ell_k}
    \exp\!\Big(-j\frac{2\pi}{\lambda}\ell_k\Big)
    \exp\!\Big(j\frac{2\pi}{\lambda}
    \tilde{\bm p}_m^T\tilde{\bm d}_k\Big).
    \label{eq:two_user_los_channel}
\end{align}
This far-field form separates the LoS channel into a directional-gain term $G_{m,k}$ and an array-induced phase term, which will be linked to channel strength and spatial separability in the ZF decomposition below.

Following the column-channel convention in the system model, we collect the conjugated scalar coefficients into $\bm h_k$, i.e., $[\bm h_k]_m=h_{k,m}^{*}$. Define the two-user channel matrix $\bm H=[\bm h_1,\bm h_2]^H\in\mathbb C^{2\times M}$, the channel strength $A_k=\|\bm h_k\|_2^2$, and the normalized channel correlation coefficient
\begin{align}
    \rho
    =
    \frac{\bm h_1^H\bm h_2}
    {\|\bm h_1\|_2\|\bm h_2\|_2}.
    \label{eq:channel_corr_def_revised}
\end{align}
Then the two-user Gram matrix is given by
\begin{align}
    \bm H\bm H^H
    =
    \begin{bmatrix}
        A_1 & \sqrt{A_1A_2}\rho\\
        \sqrt{A_1A_2}\rho^* & A_2
    \end{bmatrix}.
    \label{eq:two_user_gram_revised}
\end{align}
When $\bm H$ has full row rank, the unnormalized ZF beam directions are given by
\begin{align}
    \widetilde{\bm W}
    =
    \bm H^H(\bm H\bm H^H)^{-1}
    =
    [\tilde{\bm w}_1,\tilde{\bm w}_2],
    \label{eq:zf_pinv_revised}
\end{align}
and the normalized ZF beamformer with power allocation $\{P_k\}_{k=1}^{2}$ is $\bm w_k=\sqrt{P_k}\tilde{\bm w}_k/\|\tilde{\bm w}_k\|_2$, where $P_k\ge0$ and $\sum_{k=1}^{2}P_k\le P_{\max}$. Since $\|\tilde{\bm w}_k\|_2^2=[(\bm H\bm H^H)^{-1}]_{k,k}$, inverting \eqref{eq:two_user_gram_revised} yields
\begin{align}
    \|\tilde{\bm w}_k\|_2^2
    =
    \frac{1}{A_k(1-|\rho|^2)},
    \qquad k=1,2.
    \label{eq:zf_norm_decomposition_revised}
\end{align}
Using $\bm h_k^H \tilde{\bm w}_k=1$ and $\bm h_k^H \tilde{\bm w}_j=0$ for $j\ne k$, the SINR of user $k$ with power $P_k$ is $\text{SINR}_k = \frac{P_k}{\sigma^2 \|\tilde{\bm w}_k\|_2^2}$. Substituting \eqref{eq:zf_norm_decomposition_revised} gives the two-user ZF WSR decomposition:
\begin{align}
    R_{\rm ZF}
    =
    \sum_{k=1}^{2}
    \omega_k
    \log_2\!\Big(
    1+
    \frac{P_k}{\sigma^2}
    A_k(1-|\rho|^2)
    \Big).
    \label{eq:zf_wsr_decomposed_revised}
\end{align}
Equation~\eqref{eq:zf_wsr_decomposed_revised} shows that the ZF WSR is governed by the allocated powers $\{P_k\}$, the channel strengths $\{A_k\}$, and the spatial-separability factor $1-|\rho|^2$.

\begin{remark}
    The magnitude $|\rho|$ quantifies the degree of alignment between the two user channels. A smaller $|\rho|$ leads to a larger separability factor $1-|\rho|^2$, indicating that the ZF beamformer incurs less projection loss when suppressing inter-user interference. In the extreme case of $\rho=0$, the two channels are orthogonal, suggesting that ZF introduces no projection loss and thereby reduces to maximum-ratio transmission (MRT) in terms of beam direction.
\end{remark}

\subsection{Isotropic-Element Case: Pose-Induced Orthogonalization}
\label{subsec:p0_pose_orthogonalization}

We first consider the isotropic-element case with $p=0$. For front-side active paths, the directional gain becomes constant, i.e., $G_{m,k}=\kappa_{\max}$. Under the far-field model in \eqref{eq:two_user_los_channel}, the channel strength of user $k$ is given by
\begin{align}
    A_k=\|\bm h_k\|_2^2
    \approx
    \frac{\beta_0\kappa_{\max}M}{\ell_k^2},
    \qquad k=1,2,
    \label{eq:p0_channel_strength_revised}
\end{align}
and the normalized channel correlation in \eqref{eq:channel_corr_def_revised} reduces to
\begin{align}
    \rho(\bm R)
    =
    \frac{\nu}{M}
    \sum_{m=1}^{M}
    \exp\!\left(
    -j\frac{2\pi}{\lambda}
    \tilde{\bm p}_m^T
    \bm R^T(\bm d_2-\bm d_1)
    \right),
    \label{eq:p0_phase_corr_revised}
\end{align}
where $\nu\triangleq \exp\left(j2\pi(\ell_2-\ell_1)/\lambda\right)$ is a unit-modulus constant independent of the element index $m$.

Equations~\eqref{eq:p0_channel_strength_revised} and \eqref{eq:p0_phase_corr_revised} show that, in the isotropic-element case, the channel-strength term $A_k$ is independent of both the UPA pose and the element boresights, while the normalized channel correlation is independent of the element boresights but remains controllable by the UPA pose through the projected direction difference $\bm R^T(\bm d_2-\bm d_1)$. Therefore, UPA pose rotation can reshape the array-induced phase differences to reduce $|\rho(\bm R)|$, thereby increasing the spatial separability factor $1-|\rho(\bm R)|^2$. The following proposition gives a sufficient angular-separation condition under which this pose-dependent correlation can be driven to zero.

\begin{proposition}[Sufficient Condition for Pose-Induced Orthogonalization]
\label{prop:pose_phase_orthogonalization}
Consider a two-user far-field LoS channel with an $M_x\times M_y$ UPA composed of isotropic elements, where $M_xM_y\ge2$. Let $\alpha=\arccos(\bm d_1^T\bm d_2)$ denote the angular separation between the two users, and let $d$ and $\lambda$ be the inter-element spacing and carrier wavelength, respectively. If
\begin{align}
    2\sin\frac{\alpha}{2}
    \ge
    \min\left\{
    \frac{\lambda}{M_xd},
    \frac{\lambda}{M_yd}
    \right\},
    \label{eq:p0_pose_orth_condition_revised}
\end{align}
then there exists a UPA pose $\bm R$ that orthogonalizes the two user LoS channels, i.e., $\rho(\bm R)=0$. 
\end{proposition}

\begin{IEEEproof}
Let $\bm v=\bm d_2-\bm d_1$ and $L=\|\bm v\|_2=2\sin(\alpha/2)$. For a given pose, define the local direction-difference vector as $\tilde{\bm v}(\bm R)\triangleq\bm R^T\bm v=[\Delta_x,\Delta_y,\Delta_z]^T$, which is the user direction separation observed in the UPA LCS. Since the UPA lies in the local $x$-$y$ plane, only $\Delta_x$ and $\Delta_y$ contribute to the array-induced phase difference. For the centered $M_x\times M_y$ UPA, \eqref{eq:p0_phase_corr_revised} can be factorized as
\begin{align}
    M\rho(\bm R)
    =
    \nu
    S_{M_x}\!\left(\frac{2\pi d}{\lambda}\Delta_x\right)
    S_{M_y}\!\left(\frac{2\pi d}{\lambda}\Delta_y\right),
    \label{eq:p0_corr_factorization_revised}
\end{align}
where 
\begin{align}
    S_N(x)
    \triangleq
    \sum_{n=0}^{N-1}
    \exp\!\left(-j\left(n-\frac{N-1}{2}\right)x\right)
    =
    \frac{\sin(Nx/2)}{\sin(x/2)}.
\end{align}
Thus, $S_N(x)=0$ for $x=2\pi r/N$, $r=1,\ldots,N-1$, whenever $N\ge2$. 
Since $M_xM_y\ge2$, at least one array dimension has a nontrivial null. It is therefore sufficient to choose a valid dimension $s\in\{x,y\}$ with $M_s\ge2$ and realize
\begin{align}
    |\Delta_s|
    =
    \frac{\lambda}{M_s d},
    \label{eq:first_null_projection_revised}
\end{align}
where $M_s$ denotes the number of elements along dimension $s$. Since $\bm R$ is orthogonal, $\|\tilde{\bm v}(\bm R)\|_2=L$. With the UPA normal chosen to be orthogonal to $\bm v$ and oriented toward the two users, the full direction difference is projected onto the local $x$-$y$ plane, which yields $\Delta_z=0$. The remaining in-plane rotation can assign this length-$L$ projection to either local axis with any magnitude in $[-L,L]$. Hence, under \eqref{eq:p0_pose_orth_condition_revised}, there exists a valid dimension $s$ and an in-plane rotation satisfying \eqref{eq:first_null_projection_revised}, which makes one factor in \eqref{eq:p0_corr_factorization_revised} vanish and yields $\rho(\bm R)=0$.
\end{IEEEproof}

Proposition~\ref{prop:pose_phase_orthogonalization} provides a geometric condition under which UPA rotation can drive $\rho(\bm R)$ to zero without changing the channel-strength. It also indicates that the required angular separation decreases with the aperture size, since a larger UPA has a smaller first-null spacing in its array response. For illustration, Fig.~\ref{fig0} plots the minimum achievable correlation versus the departure angular separation $\alpha$ for $2\times2$, $3\times3$, and $4\times4$ UPAs with $d=\lambda/2$, where the minimization is evaluated by a dense search over feasible UPA poses.  All curves start from one at $\alpha=0^\circ$, since the two LoS array responses are fully aligned. As $\alpha$ increases, the pose can project the direction difference $\bm d_2-\bm d_1$ onto the planar aperture and create element-wise phase variations, thereby reducing the channel correlation. The first nulls occur at approximately $60^\circ$, $39^\circ$, and $29^\circ$ for the $2\times2$, $3\times3$, and $4\times4$ UPAs, respectively, which agrees with \eqref{eq:p0_pose_orth_condition_revised}. These results confirm that a larger aperture enables orthogonalization at a smaller user separation and that, once the condition is satisfied, pose control can remove the ZF projection loss through the factor $1-|\rho|^2$.

\begin{figure}[t]
	\begin{center}
		\includegraphics[width=0.48\textwidth]{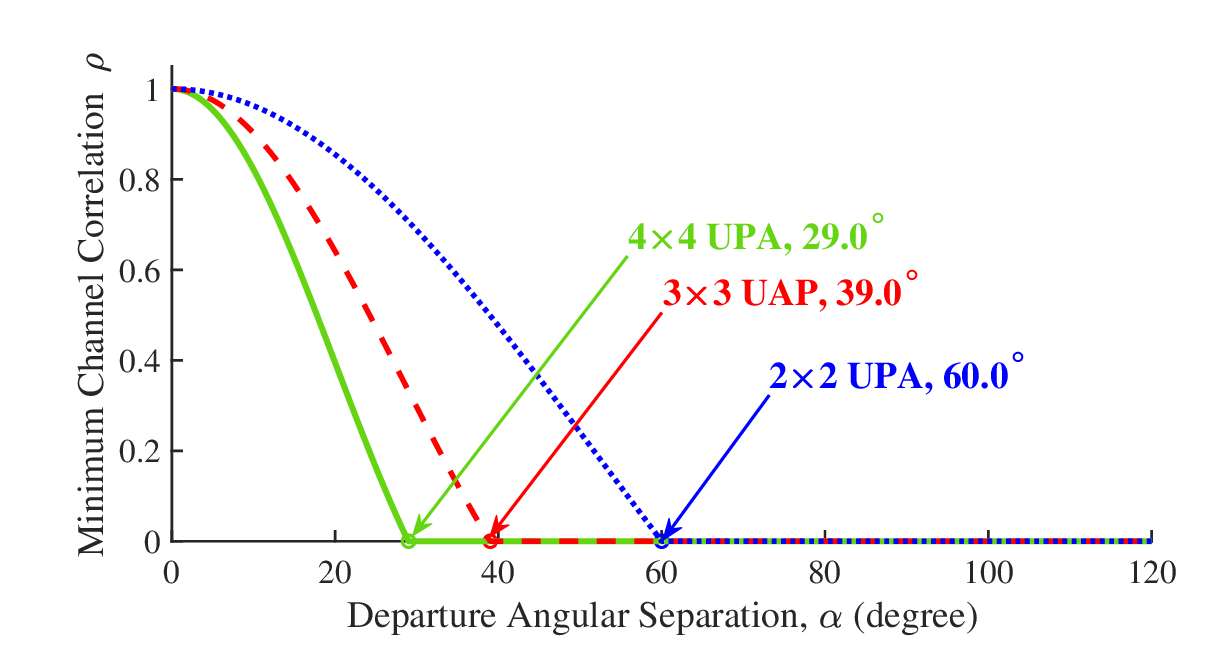}
		\caption{Minimum achievable channel correlation vs. departure separation angle with grid-search-based UPA pose selection.}
		\label{fig0}
	\end{center}
\end{figure}

\textbf{\textit{Proposed UPA-Pose Solution:}} We next discuss how to construct the UPA pose in two cases. When \eqref{eq:p0_pose_orth_condition_revised} is satisfied, an exact zero-correlation pose can be constructed in closed form. Specifically, if $M_x\ge2$ and $a_x\le L$, the first null can be enforced along the $x$-axis. Choose a front-facing UPA normal $\bm n$ orthogonal to $\bm v$, for which one simple choice is the angular-bisector normal $\bm n=(\bm d_1+\bm d_2)/\|\bm d_1+\bm d_2\|_2$. Define $\bm g=\bm v/L$ and $\bm s=\bm n\times\bm g$, and set $\theta_x=\arccos(a_x/L)$, $\bm r_x=\cos\theta_x\bm g+\sin\theta_x\bm s$, $\bm r_y=\bm n\times\bm r_x$, and $\bm r_z=\bm n$. Then $\bm R=[\bm r_x,\bm r_y,\bm r_z]\in\mathrm{SO}(3)$ satisfies $|\bm r_x^T\bm v|=a_x$, placing the $x$-axis array factor at its first null and yielding $\rho(\bm R)=0$. Similarly, if $M_y\ge2$ and $a_y\le L$, the first null can be enforced along the $y$-axis by swapping the roles of $x$ and $y$. However, when \eqref{eq:p0_pose_orth_condition_revised} is not satisfied, exact orthogonalization is not guaranteed because the available angular separation is insufficient to reach the first array-factor null. In this case, we keep the same angular-bisector normal to maximize the in-plane projection of $\bm v$, and then optimize the in-plane roll angle to minimize the residual correlation, e.g., by a one-dimensional search over $\psi$ with $\bm r_x(\psi)=\cos\psi\bm g+\sin\psi\bm s$, $\bm r_y(\psi)=\bm n\times\bm r_x(\psi)$, and $\bm r_z=\bm n$.

In summary, under isotropic elements, the UPA pose does not change the channel-strength terms but reshapes the aperture-induced phase differences between users. Its main role is therefore to improve spatial separability rather than link gain. This effect becomes stronger with a larger aperture, since a smaller user angular separation is sufficient to reach an array-factor null.

\subsection{Directional-Element Case: Pose--Boresight Coupling}
\label{subsec:pose_orientation_roles}

We next consider the directional-element case with $p>0$. Let $\tilde{\bm F}\triangleq[\tilde{\bm f}_1,\ldots,\tilde{\bm f}_M]\in \mathbb R^{3\times M}$. From the far-field LoS model in \eqref{eq:two_user_los_channel}, the channel strength of user $k$ is given by
\begin{align}
    A_k(\bm R,\tilde{\bm F})
    =
    \|\bm h_k\|_2^2
    \approx
    \frac{\beta_0}{\ell_k^2}
    \sum_{m=1}^{M}
    G_{m,k}(\bm R,\tilde{\bm F}).
    \label{eq:Ak_directional_gain}
\end{align}
Since $G_{m,k}=\kappa_{\max}[\tilde{\bm f}_m^T\bm R^T\bm d_k]_+^{2p}$, the UPA pose determines the local user direction $\bm R^T\bm d_k$, while the element boresight $\tilde{\bm f}_m$ adjusts its alignment with that direction. The same directional gains also enter the normalized correlation between the two user channels, which is given by
\begin{align}
    \rho(\bm R,\tilde{\bm F})
    =
    \frac{
    \nu
    \sum_{m=1}^{M}
    \sqrt{
    G_{m,1}
    G_{m,2}
    }
    \exp\!\Big(
    -j\frac{2\pi}{\lambda}
    \tilde{\bm p}_m^T
    (\tilde{\bm d}_2-\tilde{\bm d}_1)
    \Big)
    }
    {
    \sqrt{\sum_{m=1}^{M}G_{m,1}}
    \sqrt{\sum_{m=1}^{M}G_{m,2}}
    },
    \label{eq:weighted_corr_directional}
\end{align}
where $\nu\triangleq \exp\left(j2\pi(\ell_2-\ell_1)/\lambda\right)$. Equation~\eqref{eq:weighted_corr_directional} shows that, unlike the isotropic-element case, the channel correlation is a gain-weighted phasor sum. The UPA pose shapes the phase progression through $\tilde{\bm d}_2-\tilde{\bm d}_1=\bm R^T(\bm d_2-\bm d_1)$, while the element boresights affect the phasor weights through the directional gains. 

\begin{corollary}
\label{cor:directional_pose_orthogonalization}
Consider the two-user far-field LoS channel with directional elements, i.e., $p>0$. If the angular-separation condition in Proposition~\ref{prop:pose_phase_orthogonalization} holds, then the same pose construction can make the two directional-element LoS channels orthogonal under any common feasible boresight $\tilde{\bm f}_1=\cdots=\tilde{\bm f}_M=\tilde{\bm f}$ with nonzero gains to both users. For this construction, the channel strength of user $k$ is given by
\begin{align}
    A_k(\bm R,\tilde{\bm F})
    =
    \frac{\beta_0\kappa_{\max}M}{\ell_k^2}
    \left[
    \tilde{\bm f}^T\bm R^T\bm d_k
    \right]_+^{2p}.
    \label{eq:directional_gain_penalty_general}
\end{align}
\end{corollary}

\begin{IEEEproof}
    With a common boresight $\tilde{\bm f}$, the far-field directional gain $G_{m,k}$ is independent of the element index $m$ for each user $k$. Hence, the gain weights in \eqref{eq:weighted_corr_directional} are uniform across the array and the normalized correlation reduces to the phase-only correlation up to a unit-modulus scalar. Therefore, the first-null pose construction in Proposition~\ref{prop:pose_phase_orthogonalization} also yields $\rho=0$. The strength expression in \eqref{eq:directional_gain_penalty_general} follows from \eqref{eq:Ak_directional_gain}.
\end{IEEEproof}

Corollary~\ref{cor:directional_pose_orthogonalization} shows that, with a common boresight, the pose-induced orthogonalization condition in Proposition~\ref{prop:pose_phase_orthogonalization} still applies to directional elements, although the resulting channel strength is scaled by the directional-gain factor in \eqref{eq:directional_gain_penalty_general}. With element-wise boresight steering, however, the weights $\sqrt{G_{m,1}G_{m,2}}$ in \eqref{eq:weighted_corr_directional} generally become nonuniform, so a phase-only array-factor null created by UPA rotation may not null the true gain-weighted correlation. This highlights the pose--boresight tradeoff: the UPA pose controls aperture-domain phase separation and coarse gain alignment, whereas the element boresights refine direction-dependent gains while also reshaping the correlation weights.

\textbf{\textit{Proposed UPA-Pose Solution:}}
We construct the UPA pose through a low-dimensional geometric search over the normal tilt and the in-plane roll. To separate the pose-induced phase effect from the gain weights in \eqref{eq:weighted_corr_directional}, we use the phase-only steering vector $[\bm a(\tilde{\bm d})]_m=\exp(-j\frac{2\pi}{\lambda}\tilde{\bm p}_m^T\tilde{\bm d})$ and define the phase-only correlation as
\begin{align}
    \rho_{\rm ph}(\bm R)
    =
    \frac{1}{M}
    \bm a(\bm R^T\bm d_1)^H
    \bm a(\bm R^T\bm d_2).
    \label{eq:rho_ph_compact}
\end{align}
The corresponding phase-separation factor is $\widehat S(\bm R)=1-|\rho_{\rm ph}(\bm R)|^2$, which serves as a tractable proxy for the pose-induced spatial separability. To parameterize the pose, let $\alpha=\arccos(\bm d_1^T\bm d_2)$, $\bm v=\bm d_2-\bm d_1$, and $L=\|\bm v\|_2=2\sin(\alpha/2)$. Define the angular bisector $\bm b=(\bm d_1+\bm d_2)/\|\bm d_1+\bm d_2\|_2$ and the direction-difference unit vector $\bm g=\bm v/L$. The UPA normal is searched in the two-user plane as $\bm n(\tau)=\cos\tau\,\bm b+\sin\tau\,\bm g$, where $\tau\in\mathcal T=[-(\pi-\alpha)/2,(\pi-\alpha)/2]$. This gives the local zenith angles $\gamma_1(\tau)=|\alpha/2+\tau|$ and $\gamma_2(\tau)=|\alpha/2-\tau|$, and the available in-plane projection $L_{\perp}(\tau)=L\cos\tau$. Thus, $\tau$ captures the tradeoff between directional-gain coverage and aperture-domain phase separation. For each $\tau$, define $S_{\rm ph}^{\star}(\tau)\triangleq\max_{\psi}\big(1-|\rho_{\rm ph}(\bm R(\tau,\psi))|^2\big)$ as the best phase-separation factor after optimizing the in-plane roll angle $\psi$. We choose the normal tilt by
\begin{align}
    \tau^{\star}
    \in
    \arg\max_{\tau\in\mathcal T}
    &\sum_{k=1}^{2}
    \omega_k
    \log_2\!\Big(
    1+
    \frac{P_k}{\sigma^2}
    \frac{\beta_0\kappa_{\max}M}{\ell_k^2} \notag \\
    &\quad 
    \cos^{2p}
    \left(
    [\gamma_k(\tau)-\theta_{\max}]_+
    \right)
    S_{\rm ph}^{\star}(\tau)
    \Big).
    \label{eq:tau_design_compact}
\end{align}
Here, the cosine term approximates the best directional alignment allowed by the spherical-cap steering constraint, while $S_{\rm ph}^{\star}(\tau)$ measures the phase-separation capability for the given normal tilt. After obtaining $\tau^\star$, set $\bm n^\star=\bm n(\tau^\star)$ and $L_\perp^\star=L\cos\tau^\star$. Let $a_x=\lambda/(M_xd)$ and $a_y=\lambda/(M_yd)$ denote the first-null projections along the two aperture dimensions. If $L_\perp^\star\ge\min\{a_x,a_y\}$, a first-null roll can be selected in closed form. For example, if the $x$-axis null is used, $\psi^\star=\arccos(a_x/L_\perp^\star)$ yields $|\Delta_x|=a_x$; otherwise, the $y$-axis null can be enforced similarly. If no first-null roll is available, $\psi^\star$ is obtained by a one-dimensional search that minimizes $|\rho_{\rm ph}(\bm R)|$. Finally, define $\bm u^\star=(\bm I-\bm n^\star(\bm n^\star)^T)\bm v/\|(\bm I-\bm n^\star(\bm n^\star)^T)\bm v\|_2$, $\bm s^\star=\bm n^\star\times\bm u^\star$, $\bm r_x=\cos\psi^\star\bm u^\star+\sin\psi^\star\bm s^\star$, $\bm r_y=\bm n^\star\times\bm r_x$, and $\bm r_z=\bm n^\star$. The resulting pose $\bm R^\star=[\bm r_x,\bm r_y,\bm r_z]$ satisfies $\bm R^\star\in\mathrm{SO}(3)$ and $\bm R^\star\bm e_z=\bm n^\star$.

\begin{remark}
    The pose search provides coverage-aware phase separation: the normal tilt $\tau$ trades feasible directional-gain alignment against the available aperture projection, while the roll $\psi^\star$ reduces the resulting phase-domain correlation.
\end{remark}

\textbf{\textit{Proposed Element-Boresight Solution:}}
With $\bm R^\star$ fixed, the element boresights are selected based on the element-wise participation in the phase-only ZF directions. Let $\tilde{\bm d}_k^\star=(\bm R^\star)^T\bm d_k$ and $\bar{\bm H}=[\bm a(\tilde{\bm d}_1^\star),\bm a(\tilde{\bm d}_2^\star)]^H$. When $\bar{\bm H}$ has full row rank, we compute $\bar{\bm W}=\bar{\bm H}^H(\bar{\bm H}\bar{\bm H}^H)^{-1}$ and denote $\bar{\bm W}=[\bar{\bm w}_1,\bar{\bm w}_2]$. Since $|[\bar{\bm w}_k]_m|^2$ measures the contribution of element $m$ to the phase-only ZF direction of user $k$, the boresight of element $m$ is chosen as
\begin{align}
    \tilde{\bm f}_m^\star
    =
    \Pi_{\mathcal F}
    \left(
    \frac{
    \sum_{k=1}^{2}
    \omega_k
    |[\bar{\bm w}_k]_m|^2
    \tilde{\bm d}_k^\star
    }
    {
    \left\|
    \sum_{k=1}^{2}
    \omega_k
    |[\bar{\bm w}_k]_m|^2
    \tilde{\bm d}_k^\star
    \right\|_2
    }
    \right),
    \quad \forall m,
    \label{eq:orientation_seq_compact}
\end{align}
where $\Pi_{\mathcal F}(\cdot)$ denotes the projection onto $\mathcal F$, and the current feasible boresight is retained if the numerator is zero. This rule aligns each element with a ZF-participation-weighted local user direction, and the final ZF precoders are computed from the effective channels induced by $\bm R^\star$ and $\{\tilde{\bm f}_m^\star\}$.

In summary, with directional elements, UPA pose and element boresights jointly shape two coupled channel features: spatial separability and channel strength. The UPA pose changes the array-induced phase progression and the local user directions, thereby affecting inter-user separability and coarse directional-gain alignment. Element-boresight steering does not change the propagation phases, but refines the direction-dependent gains, which determine the channel strengths $A_k$ and the gain weights in $\rho(\bm R,\tilde{\bm F})$. Hence, reducing the phase-only correlation alone is insufficient, and joint pose--boresight design is required to balance spatial separability and channel-strength enhancement.

\section{Joint Beamforming and Array Geometry Optimization}
\label{sec:AO_solution}

The joint WSR maximization in \eqref{prob:WSR_MU} is highly nonconvex due to the coupling among the transmit beamformers, UPA pose, and element boresights. To obtain an efficient and implementable solution, we develop an alternating optimization (AO) framework with three structured blocks: WMMSE-based beamforming update, rotation-manifold-based UPA pose optimization, and spherical-cap-constrained element-boresight refinement. Each block update preserves feasibility and does not decrease the WSR objective, thereby yielding a nondecreasing objective sequence.

\subsection{Optimization for Transmit Beamformers}

For fixed UPA pose $\bm R$ and element boresights $\{\tilde{\bm f}_m\}$, the channels $\{\bm h_k\}$ are fixed and \eqref{prob:WSR_MU} reduces to the conventional WSR beamforming problem under a sum-power constraint. We solve this problem using the standard WMMSE reformulation \cite{5756489}. For each user $k$, introduce a scalar receive equalizer $u_k\in\mathbb C$ for estimating $s_k$ from $y_k$. The mean-square error (MSE) is defined as $e_k=\mathbb E[|u_k y_k-s_k|^2]$, which can be written as
\begin{align}
    e_k
    =
    |u_k|^2
    \big(\sum_{i=1}^K |\bm h_k^H\bm w_i|^2+\sigma^2\big)
    -2\Re\{u_k\bm h_k^H\bm w_k\}+1.
    \label{eq:mse_k}
\end{align}
With a positive auxiliary weight $v_k>0$ associated with each MSE, the WMMSE problem can be written as
\begin{align}
    \min_{\{u_k\},\{v_k\},\{\bm w_k\}} \,\,
    &\sum_{k=1}^K \omega_k\big(v_k e_k-\ln v_k\big) \\
    \mathrm{s.t.}\qquad
    &\sum_{k=1}^K\|\bm w_k\|_2^2\le P_{\max}.
    \label{prob:wmmse}
\end{align}
For given beamformers, the optimal equalizers and auxiliary weights are given by
\begin{align}
    u_k
    =
    \frac{(\bm h_k^H\bm w_k)^*}
    {\sum_{i=1}^K |\bm h_k^H\bm w_i|^2+\sigma^2},
    \qquad
    v_k=e_k^{-1}.
    \label{eq:u_v_update}
\end{align}
For given $\{u_k,v_k\}$, the beamformers are updated by
\begin{align}
    \bm w_k
    =
    \left(
    \sum_{j=1}^K
    \omega_j v_j |u_j|^2 \bm h_j\bm h_j^H
    +\mu\bm I
    \right)^{-1}
    \omega_k v_k u_k^*\bm h_k,
    \label{eq:w_update}
\end{align}
where $\mu\ge0$ is adjusted by bisection to satisfy the sum-power constraint. The above updates are repeated until convergence.

\subsection{Optimization for UPA Pose}

For fixed beamformers $\{\bm w_k\}$ and element boresights $\{\tilde{\bm f}_m\}$, the UPA pose optimization subproblem can be written over the rotation manifold $\mathrm{SO}(3)$ as
\begin{align}
    \max_{\bm R\in\mathrm{SO}(3)}\quad
    &\Phi(\bm R)
    \triangleq
    \sum_{k=1}^K \omega_k r_k(\bm R)
    \label{prob:pose_R}\\
    \mathrm{s.t.}\quad
    &(\bm R\bm e_z)^T(\bm q_k-\bm p_0)\ge0,\quad \forall k.
    \label{prob:pose_R_front}
\end{align}
Here, $\bm R$ affects the channel through both the element positions $\bm p_m=\bm p_0+\bm R\tilde{\bm p}_m$ and the global boresights of each element $\bm f_m=\bm R\tilde{\bm f}_m$. Problem \eqref{prob:pose_R} is nonconvex due to the rotation manifold and the nonlinear dependence of the WSR on the pose-dependent channel.

We update the UPA pose intrinsically on $\mathrm{SO}(3)$ via a local incremental rotation, which preserves the rotation-matrix constraints by construction. Specifically, at the $t$-th iteration, given the current feasible pose $\bm R^{(t)}$ and a small search direction $\bm\delta\in\mathbb R^3$, the candidate pose is generated as
\begin{align}
    \bm R^+(\bm\delta)
    =
    \bm R^{(t)}\exp([\bm\delta]_\times),
    \label{eq:pose_update_so3_new}
\end{align}
where $[\bm\delta]_\times$ is the skew-symmetric matrix satisfying $[\bm\delta]_\times\bm a=\bm\delta\times\bm a$ for any $\bm a\in\mathbb R^3$, with $\times$ denoting the vector cross product. For a sufficiently small local rotation increment $\bm\delta$, the exponential map admits the first-order approximation
\begin{align}
    \exp([\bm\delta]_\times)
    =
    \bm I+[\bm\delta]_\times
    +
    \mathcal O(\|\bm\delta\|_2^2).
\end{align}
Thus, the UPA normal vector in \eqref{prob:pose_R_front} can be locally approximated as
\begin{align}
    \bm R^+(\bm\delta)\bm e_z
    \approx
    \bm R^{(t)}\bm e_z
    +\bm R^{(t)}[\bm\delta]_\times\bm e_z.
    \label{eq:lin_n_new}
\end{align}
The front half-space constraint admits the local approximation
\begin{align}
    \big(\bm R^{(t)}\bm e_z
    +\bm R^{(t)}[\bm\delta]_\times\bm e_z\big)^T(\bm q_k-\bm p_0)
    \ge 0,\quad \forall k.
    \label{eq:lin_front_new}
\end{align}
By defining $\bm b_k^{(t)}\triangleq\bm e_z\times(\bm R^{(t)})^T(\bm q_k-\bm p_0)$, we have $\big(\bm R^{(t)}[\bm\delta]_\times\bm e_z\big)^T(\bm q_k-\bm p_0) = (\bm b_k^{(t)})^T\bm\delta$, and hence the linearized constraint can be written as
\begin{align}
    (\bm R^{(t)}\bm e_z)^T(\bm q_k-\bm p_0)+(\bm b_k^{(t)})^T\bm\delta\ge0,\quad \forall k.
    \label{eq:lin_front_bk}
\end{align}
Next, we construct a local first-order ascent model of the objective on $\mathrm{SO}(3)$. To this end, we define the local objective with respect to an incremental rotation $\bm\delta$ as
\begin{align}
    \varphi_t(\bm\delta) = \Phi\big(\bm R^{(t)}\exp([\bm\delta]_\times)\big),
\end{align}
where $\bm\delta=\bm 0$ corresponds to the current pose $\bm R^{(t)}$. The local ascent direction is determined by the right-trivialized gradient of $\Phi(\bm R)$ at $\bm R^{(t)}$, defined as
\begin{align}
    \bm g_R^{(t)}
    =
    \nabla_{\bm\delta}\varphi_t(\bm\delta)\big|_{\bm\delta=\bm 0}.
\end{align}
Since $\Phi(\bm R)=\sum_{k=1}^K\omega_k r_k(\bm R)$, the $\ell$-th component of $\bm g_R^{(t)}$ can be written as
\begin{align}
    [\bm g_R^{(t)}]_\ell
    =
    \sum_{k=1}^K
    \omega_k
    \frac{\partial r_k}{\partial\delta_\ell},
    \quad \ell=1,2,3.
    \label{eq:pose_grad_wsr}
\end{align}
It remains to compute the rate derivative with respect to each local rotation component. Define $z_{k,i}=\bm h_k^H\bm w_i$, $T_k=\sum_{i=1}^K|z_{k,i}|^2+\sigma^2$, and $D_k=\sum_{i\ne k}|z_{k,i}|^2+\sigma^2$. Then, by differentiating the rate expression and applying the chain rule through the pose-dependent channel, we obtain
\begin{align}
    \frac{\partial r_k}{\partial \delta_\ell}
    =
    \frac{2}{\ln 2}
    \Bigg[
    &\frac{
    \Re\!\left\{
    \sum_{i=1}^K
    z_{k,i}^*
    \left(
    \frac{\partial \bm h_k}{\partial \delta_\ell}
    \right)^H
    \bm w_i
    \right\}
    }{T_k}
    \notag\\
    &-
    \frac{
    \Re\!\left\{
    \sum_{i\ne k}
    z_{k,i}^*
    \left(
    \frac{\partial \bm h_k}{\partial \delta_\ell}
    \right)^H
    \bm w_i
    \right\}
    }{D_k}
    \Bigg].
    \label{eq:rate_grad_pose}
\end{align}
The remaining task is to evaluate the channel derivative $\partial\bm h_k/\partial\delta_\ell$, which is obtained component-wise as
\begin{align}
    \frac{\partial \bm h_k}{\partial \delta_\ell}
    =
    \left[
    \frac{\partial[\bm h_k]_1}{\partial \delta_\ell},
    \ldots,
    \frac{\partial[\bm h_k]_M}{\partial \delta_\ell}
    \right]^T,
\end{align}
where the derivative of the $m$-th entry is given by
\begin{align}
    \frac{\partial[\bm h_k]_m}{\partial \delta_\ell}
    =
    \frac{\partial (h_{k,m}^{\mathrm{LoS}})^*}{\partial \delta_\ell}
    +
    \sum_{q=1}^Q
    \frac{\partial (h_{k,m}^{(q)})^*}{\partial \delta_\ell}.
    \label{eq:channel_grad_pose}
\end{align}
To unify the LoS and NLoS derivatives, consider a generic BS-side path from element $m$ to a target point $\bm a$, where $\bm a=\bm q_k$ for the LoS link and $\bm a=\bm s_q$ for the BS-to-cluster segment of the $q$-th NLoS path. Define $r_{m,a}=\|\bm a-\bm p_m\|$, $\bm d_{m,a}=(\bm a-\bm p_m)/r_{m,a}$, and $c_{m,a}=\bm f_m^T\bm d_{m,a}$. Let $h$ denote the corresponding path coefficient in $[\bm h_k]_m$. For a front-side active path with $c_{m,a}>0$, its derivative with respect to the local rotation component $\delta_\ell$ can be written as
\begin{align}
    \frac{\partial h}{\partial \delta_\ell}
    =
    h
    \left(
    \frac{p}{c_{m,a}}\frac{\partial c_{m,a}}{\partial \delta_\ell}
    -
    \frac{1}{r_{m,a}}\frac{\partial r_{m,a}}{\partial \delta_\ell}
    +
    j\frac{2\pi}{\lambda}
    \frac{\partial r_{m,a}}{\partial \delta_\ell}
    \right),
    \label{eq:path_grad_pose}
\end{align}
where
\begin{align}
    \frac{\partial r_{m,a}}{\partial \delta_\ell}
    =
    -\bm d_{m,a}^T
    \bm R^{(t)}(\bm e_\ell\times\tilde{\bm p}_m),
    \label{eq:r_grad_pose}
\end{align}
and
\begin{align}
    \frac{\partial c_{m,a}}{\partial \delta_\ell}
    &=
    \left[
    \bm R^{(t)}(\bm e_\ell\times\tilde{\bm f}_m)
    \right]^T
    \bm d_{m,a}
    \notag\\
    &\quad
    -
    \frac{1}{r_{m,a}}
    \bm f_m^T
    \left(
    \bm I-\bm d_{m,a}\bm d_{m,a}^T
    \right)
    \bm R^{(t)}(\bm e_\ell\times\tilde{\bm p}_m).
    \label{eq:c_grad_pose}
\end{align}
For paths with $c_{m,a}\le0$, the directional gain is zero under the adopted half-space cosine pattern, and the corresponding path derivative is set to zero.

Based on the analytical gradient and the local constraint approximation, the incremental rotation direction is obtained by solving
\begin{align}
    \max_{\bm\delta\in\mathbb R^3}\quad
    &(\bm g_R^{(t)})^T\bm\delta
    \label{prob:delta_new}\\
    \mathrm{s.t.}\quad
    & (\bm R^{(t)}\bm e_z)^T(\bm q_k-\bm p_0)+(\bm b_k^{(t)})^T\bm\delta\ge0,\quad \forall k.
    \label{prob:delta_front_new}\\
    &\|\bm\delta\|_\infty\le\delta_{\max}.
    \label{prob:delta_tr_new}
\end{align}
The trust-region constraint in \eqref{prob:delta_tr_new} restricts the local rotation magnitude and helps maintain the accuracy of the first-order approximation. Since \eqref{prob:delta_new} is only a three-dimensional linear program with linearized visibility and box constraints, it can be solved efficiently by standard linear programming solvers.

Given the search direction $\bm\delta$, the step size $\rho$ is selected by a feasible Armijo backtracking procedure. Starting from $\rho=\rho_0$, we shrink $\rho\leftarrow\xi\rho$ with $\xi\in(0,1)$ until the candidate update satisfies both
\begin{align}
    \Phi(\bm R^+)
    \ge
    \Phi(\bm R^{(t)})
    +
    c\rho(\bm g_R^{(t)})^T\bm\delta,
    \label{eq:armijo_pose_new}
\end{align}
and the original front half-space constraints
\begin{align}
    (\bm R^+\bm e_z)^T(\bm q_k-\bm p_0)\ge0,\quad \forall k.
    \label{eq:pose_exact_feas_new}
\end{align}
Since every accepted update satisfies \eqref{eq:pose_exact_feas_new} and \eqref{eq:armijo_pose_new}, the UPA pose remains feasible and the WSR is monotonically nondecreasing during the pose update. The complete algorithm for UPA pose optimization is summarized in Algorithm~\ref{alg:pose_so3}.

\begin{algorithm}[t]\small
\caption{Optimization for UPA Pose}
\label{alg:pose_so3}
\begin{algorithmic}[1]
\REQUIRE $\{\bm w_k\}$, $\{\tilde{\bm f}_m\}$, feasible initial pose $\bm R^{(0)}$, and algorithmic parameters $\delta_{\max}$, $\rho_0$, $\rho_{\min}$, $\xi$, $c$, $\varepsilon_{\rm obj}$, and $T_{\mathrm{R}}$.
\STATE Initialize $\bm R\leftarrow\bm R^{(0)}$ and $\Phi_{\rm cur}\leftarrow\Phi(\bm R)$.
\FOR{$t=0,1,\ldots,T_{\mathrm{R}}-1$}
    \STATE Compute $\bm g_R^{(t)}$ by \eqref{eq:pose_grad_wsr}--\eqref{eq:c_grad_pose}.
    \STATE Obtain the local rotation direction $\bm\delta$ by solving \eqref{prob:delta_new}.
    \STATE Set $\Phi_{\rm old}\leftarrow\Phi_{\rm cur}$.
    \STATE Find a $\rho\in\{\rho_0,\xi\rho_0,\xi^2\rho_0,\ldots\}$ such that $\rho\ge\rho_{\min}$ and the candidate $\bm R^+=\bm R\exp(\rho[\bm\delta]_\times)$ satisfies \eqref{eq:armijo_pose_new} and \eqref{eq:pose_exact_feas_new}.
    \IF{no such $\rho$ is found}
        \STATE \textbf{break}
    \ENDIF
    \STATE Update $\bm R\leftarrow\bm R^+$ and $\Phi_{\rm cur}\leftarrow\Phi(\bm R)$.
    \IF{$|\Phi_{\rm cur}-\Phi_{\rm old}|<\varepsilon_{\rm obj}$}
        \STATE \textbf{break}
    \ENDIF
\ENDFOR
\STATE \textbf{return} $\bm R$.
\end{algorithmic}
\end{algorithm}

\subsection{Optimization for Element Boresights}

For fixed beamformers $\{\bm w_k\}$ and UPA pose $\bm R$, we optimize the element boresights $\{\tilde{\bm f}_m\}$ in the UPA LCS. Let $\tilde{\bm F}\triangleq[\tilde{\bm f}_1,\ldots,\tilde{\bm f}_M]\in \mathbb R^{3\times M}$. The element-boresight subproblem is
\begin{align}
    \max_{\tilde{\bm F}}\quad
    & \Psi(\tilde{\bm F})
    \triangleq
    \sum_{k=1}^K\omega_k r_k(\tilde{\bm F})
    \label{prob:ori}\\
    \mathrm{s.t.}\quad
    & \tilde{\bm f}_m\in\mathcal F,\quad \forall m,
    \label{const:ori_cap}
\end{align}
where $\mathcal F$ is the spherical-cap feasible set defined in \eqref{eq:ori_feasible_set}. With $\bm R$ fixed, $\tilde{\bm F}$ affects the rates through the global boresights $\bm f_m=\bm R\tilde{\bm f}_m$ and the directional gains in \eqref{eq:RA_pattern_MU}. Problem \eqref{prob:ori} is nonconvex due to the nonconcave WSR objective and the nonconvex spherical-cap constraints. To this end, we adopt a Frank--Wolfe (FW)-type manifold-aware ascent scheme.

Let $\bm g_m\triangleq\nabla_{\tilde{\bm f}_m}\Psi(\tilde{\bm F})$ denote the Euclidean gradient with respect to $\tilde{\bm f}_m$. Using the previously defined quantities $z_{k,i}=\bm h_k^H\bm w_i$, $T_k=\sum_{i=1}^K|z_{k,i}|^2+\sigma^2$, and $D_k=\sum_{i\ne k}|z_{k,i}|^2+\sigma^2$, the gradient can be expressed as
\begin{align}
    \bm g_m
    =
    \sum_{k=1}^K
    \frac{2\omega_k}{\ln 2}
    \Bigg[
    &\frac{
    \Re\!\left\{
    \sum_{i=1}^K
    z_{k,i}^*
    (\bm J_{k,m}^{\rm ori})^H\bm w_i
    \right\}
    }{T_k}
    \notag\\
    &\quad -
    \frac{
    \Re\!\left\{
    \sum_{i\ne k}
    z_{k,i}^*
    (\bm J_{k,m}^{\rm ori})^H\bm w_i
    \right\}
    }{D_k}
    \Bigg],
    \label{eq:ori_grad_gm}
\end{align}
where $\bm J_{k,m}^{\rm ori}\triangleq\partial\bm h_k/\partial\tilde{\bm f}_m^T$. Since $\tilde{\bm f}_m$ only affects the $m$-th transmit element, $\bm J_{k,m}^{\rm ori}$ is sparse and its only nonzero row is given by
\begin{align}
    [\bm J_{k,m}^{\rm ori}]_{m,:}
    =
    \left(
    \frac{\partial (h_{k,m}^{\mathrm{LoS}})^*}{\partial \tilde{\bm f}_m}
    +
    \sum_{q=1}^Q
    \frac{\partial (h_{k,m}^{(q)})^*}{\partial \tilde{\bm f}_m}
    \right)^T .
    \label{eq:ori_J_row}
\end{align}
For an active LoS path with $c_{m,k}\triangleq\bm f_m^T\bm d_{m,k}>0$, the derivative is given by
\begin{align}
    \frac{\partial (h_{k,m}^{\mathrm{LoS}})^*}{\partial \tilde{\bm f}_m}
    =
    \frac{p(h_{k,m}^{\mathrm{LoS}})^*}{c_{m,k}}
    \bm R^T\bm d_{m,k}.
    \label{eq:ori_los_grad}
\end{align}
Similarly, for an active NLoS BS-to-cluster segment with $c_{m,q}\triangleq\bm f_m^T\bm d_{m,q}>0$, we have
\begin{align}
    \frac{\partial (h_{k,m}^{(q)})^*}{\partial \tilde{\bm f}_m}
    =
    \frac{p(h_{k,m}^{(q)})^*}{c_{m,q}}
    \bm R^T\bm d_{m,q}.
    \label{eq:ori_nlos_grad}
\end{align}
For inactive paths with $c_{m,k}\le0$ or $c_{m,q}\le0$, the corresponding derivative is set to zero due to the half-space radiation pattern. To define a local ascent score on the unit sphere, we project the Euclidean gradient onto the tangent space at $\tilde{\bm f}_m$:
\begin{align}
    \bm g_m^{\rm tan}
    =
    \left(\bm I-\tilde{\bm f}_m\tilde{\bm f}_m^T\right)\bm g_m,
    \label{eq:ori_tangent}
\end{align}
where $(\bm g_m^{\rm tan})^T\tilde{\bm f}_m=0$. The local linearized ascent from $\tilde{\bm f}_m$ to a feasible point $\tilde{\bm f}\in\mathcal F$ is proportional to $(\bm g_m^{\rm tan})^T(\tilde{\bm f}-\tilde{\bm f}_m)$. Hence, the FW-type linear oracle is given by
\begin{align}
    \tilde{\bm f}_m^{\rm FW}
    \in
    \arg\max_{\tilde{\bm f}\in\mathcal F}
    (\bm g_m^{\rm tan})^T\tilde{\bm f}.
    \label{eq:ori_oracle}
\end{align}
If $\|\bm g_m^{\rm tan}\|_2=0$, we set $\tilde{\bm f}_m^{\rm FW}=\tilde{\bm f}_m$. Otherwise, let $\hat{\bm g}_m=\bm g_m^{\rm tan}/\|\bm g_m^{\rm tan}\|_2$ and $c_\theta=\cos\theta_{\max}$. When the horizontal projection is nonzero, define
\begin{align}
    \bm e_m^\perp
    =
    \frac{
    (\bm I-\bm e_z\bm e_z^T)\hat{\bm g}_m
    }{
    \|(\bm I-\bm e_z\bm e_z^T)\hat{\bm g}_m\|_2
    } .
    \label{eq:horizontal_projection}
\end{align}
Then the oracle solution is obtained as
\begin{align}
    \tilde{\bm f}_m^{\rm FW}
    =
    \begin{cases}
    \hat{\bm g}_m,
    & \hat{\bm g}_m^T\bm e_z\ge c_\theta,\\[1mm]
    c_\theta\bm e_z+\sqrt{1-c_\theta^2}\,\bm e_m^\perp,
    & \hat{\bm g}_m^T\bm e_z<c_\theta.
    \end{cases}
    \label{eq:ori_oracle_closed_form}
\end{align}
If the denominator in \eqref{eq:horizontal_projection} is zero in the second case, any boundary point of $\mathcal F$ satisfying $\tilde{\bm f}^T\bm e_z=c_\theta$ is optimal.

With the obtained $\{\tilde{\bm f}_m^{\rm FW}\}$, define the FW gap as
\begin{align}
    \Delta_{\rm FW}(\tilde{\bm F})
    \triangleq
    \sum_{m=1}^M
    (\bm g_m^{\rm tan})^T(\tilde{\bm f}_m^{\rm FW}-\tilde{\bm f}_m).
    \label{eq:fw_gap}
\end{align}
If $\Delta_{\rm FW}(\tilde{\bm F})$ is below a small threshold, the boresight update terminates. Otherwise, we update along $\tilde{\bm f}_m^{\rm FW}-\tilde{\bm f}_m$ with step size $\rho\in(0,1]$. To preserve feasibility, we use a normalized convex-combination update followed by a projection onto $\mathcal F$:
\begin{align}
    \bm v_m(\rho)
    &=
    (1-\rho)\tilde{\bm f}_m+\rho \tilde{\bm f}_m^{\rm FW},\\
    \tilde{\bm f}_m^+(\rho)
    &=
    \Pi_{\mathcal F}
    \left(
    \frac{\bm v_m(\rho)}{\|\bm v_m(\rho)\|_2}
    \right),
    \label{eq:ori_update}
\end{align}
where $\Pi_{\mathcal F}(\cdot)$ denotes the projection onto $\mathcal F$. The step size is selected by Armijo backtracking:
\begin{align}
    \Psi(\tilde{\bm F}^+(\rho))
    \ge
    \Psi(\tilde{\bm F})
    +
    c\rho\Delta_{\rm FW}(\tilde{\bm F}),
    \label{eq:ori_armijo}
\end{align}
where $c\in(0,1)$. Starting from $\rho_0$, we shrink $\rho\leftarrow\xi\rho$ with $\xi\in(0,1)$ until \eqref{eq:ori_armijo} holds. By construction, the accepted update preserves the spherical-cap constraints and yields a nondecreasing WSR sequence. The element-boresight update is summarized in Algorithm~\ref{alg:ori_only}.

\begin{algorithm}[t]\small
\caption{Optimization for Element Boresights}
\label{alg:ori_only}
\begin{algorithmic}[1]
\REQUIRE $\{\bm w_k\}$, $\bm R$, $\theta_{\max}$, feasible initial boresights $\tilde{\bm F}^{(0)}$, and algorithmic parameters $\rho_0$, $\rho_{\min}$, $\xi$, $c$, $\varepsilon_{\rm FW}$, $\varepsilon_{\rm obj}$, and $T_{\mathrm{F}}$.
\STATE Initialize $\tilde{\bm F}\leftarrow\tilde{\bm F}^{(0)}$ and $\Psi_{\rm cur}\leftarrow\Psi(\tilde{\bm F})$.
\FOR{$t=0,1,\ldots,T_{\mathrm{F}}-1$}
    \STATE Compute $\bm g_m$ by \eqref{eq:ori_grad_gm} and $\bm g_m^{\rm tan}$ by \eqref{eq:ori_tangent}, $\forall m$.
    \STATE Obtain $\tilde{\bm f}_m^{\rm FW}$ by \eqref{eq:ori_oracle_closed_form}, $\forall m$, and compute $\Delta_{\rm FW}$ by \eqref{eq:fw_gap}.
    \IF{$\Delta_{\rm FW}\le\varepsilon_{\rm FW}$}
        \STATE \textbf{break}
    \ENDIF
    \STATE Set $\Psi_{\rm old}\leftarrow\Psi_{\rm cur}$.
    \STATE Find a step size $\rho\in\{\rho_0,\xi\rho_0,\xi^2\rho_0,\ldots\}$ such that $\rho\ge\rho_{\min}$ and $\tilde{\bm F}^{+}(\rho)$ in \eqref{eq:ori_update} satisfies \eqref{eq:ori_armijo}.
    \IF{no such $\rho$ is found}
        \STATE \textbf{break}
    \ENDIF
    \STATE Update $\tilde{\bm F}\leftarrow\tilde{\bm F}^{+}(\rho)$ and $\Psi_{\rm cur}\leftarrow\Psi(\tilde{\bm F})$.
    \IF{$|\Psi_{\rm cur}-\Psi_{\rm old}|<\varepsilon_{\rm obj}$}
        \STATE \textbf{break}
    \ENDIF
\ENDFOR
\STATE \textbf{return} $\tilde{\bm F}$.
\end{algorithmic}
\end{algorithm}

\subsection{Overall AO Algorithm and Complexity Analysis}
\label{anaAO}
The overall AO procedure is summarized in Algorithm~\ref{alg:overall_ao}. Starting from feasible $\bm R$ and $\tilde{\bm F}$, the algorithm cyclically updates the beamformers, the UPA pose, and the element boresights by the WMMSE method, Algorithm~\ref{alg:pose_so3}, and Algorithm~\ref{alg:ori_only}, respectively. After each block update, the channels and the WSR objective are recomputed. Since each accepted block update preserves feasibility and does not decrease the WSR, the objective sequence generated by Algorithm~\ref{alg:overall_ao} is monotonically nondecreasing. Moreover, under finite transmit power and bounded channel gains, the WSR is upper bounded, and thus the objective sequence converges.

\begin{remark}
    \label{rem:ao_special_cases}
    The proposed AO framework also covers several special cases. If the element boresights $\{\tilde{\bm f}_m\}$ are fixed, it reduces to a pose-only design that alternates between WMMSE beamforming and UPA pose optimization. If the UPA pose $\bm R$ is fixed, it reduces to a boresight-only design that alternates between WMMSE beamforming and element boresight optimization. Moreover, as discussed in Remark~\ref{rem:channel_geometric_roles}, when $p=0$, the element boresights provide no angle-dependent gain adjustment within the visible front half-space, while the UPA pose can still reshape the effective channels through the element positions and relative phases.
\end{remark}

\begin{algorithm}[t]\small
\caption{Overall AO-Based Algorithm}
\label{alg:overall_ao}
\begin{algorithmic}[1]
\REQUIRE Feasible UPA pose $\bm R^{(0)}$, feasible element boresights $\tilde{\bm F}^{(0)}$, and algorithmic parameters $T_{\mathrm{AO}}$ and $\varepsilon_{\rm AO}$.
\STATE Initialize $t\leftarrow0$, $\bm R\leftarrow\bm R^{(0)}$, and
$\tilde{\bm F}\leftarrow\tilde{\bm F}^{(0)}$.
\STATE Compute the initial WSR objective $\mathcal R^{(0)}$.
\REPEAT
    \STATE Update $\{\bm w_k\}$ by solving the WMMSE problem
    in \eqref{prob:wmmse}.
    \STATE Update the UPA pose $\bm R$ by Algorithm~\ref{alg:pose_so3}.
    \STATE Update the element boresights $\tilde{\bm F}$ by
    Algorithm~\ref{alg:ori_only}.
    \STATE Compute the WSR objective
    $\mathcal R^{(t+1)}$.
    \STATE $t\leftarrow t+1$.
    \UNTIL{$t\ge T_{\mathrm{AO}}$ or $|\mathcal R^{(t)}-\mathcal R^{(t-1)}| < \varepsilon_{\rm AO} $.}
\STATE \textbf{return} $\{\bm w_k\}$, $\bm R$, and $\tilde{\bm F}$.
\end{algorithmic}
\end{algorithm}

Computational complexity: In the beamforming block, each WMMSE iteration is dominated by computing the user-beam inner products, forming the weighted channel covariance matrix, and solving a common $M$-dimensional linear system, with complexity $\mathcal O(K^2M+KM^2+M^3)$. In the UPA-pose block, only a three-dimensional rotation increment is optimized, and the main cost comes from computing the pose gradient and evaluating the WSR, with complexity about $\mathcal O(K^2M+KM(Q+1))$ per pose iteration. In the element-boresight block, all $M$ boresight vectors are updated. By exploiting the sparsity of $\bm J_{k,m}^{\rm ori}$, computing the element-wise gradients and the closed-form FW-type oracle has complexity about $\mathcal O(K^2M+KM(Q+1))$ per boresight iteration, with a larger constant due to the element-wise variables. Let $T_{\rm AO}$, $T_{\rm W}$, $T_{\rm R}$, and $T_{\rm F}$ denote the numbers of outer AO iterations, WMMSE iterations, UPA-pose iterations, and element-boresight iterations, respectively. The overall complexity can be approximately characterized as $\mathcal O\big(T_{\rm AO}[T_{\rm W}(K^2M+KM^2+M^3) +(T_{\rm R}+T_{\rm F})(K^2M+KM(Q+1))]\big)$.

\section{Simulation Results}
In this section, numerical examples are provided to validate the effectiveness of the proposed algorithms and the geometric interpretation. Unless otherwise stated, the carrier frequency is $f_c=2.4$~GHz, the inter-element spacing is $d=\lambda/2$, the receiver noise power is $-80$~dBm, and equal user weights $\omega_k=1$ are adopted. The BS is located at height $h_{\mathrm{BS}}=50$~m, while the users are located on the horizontal plane $z=h_{\mathrm{UE}}=1.5$~m. We first consider a controlled two-user LoS setup to validate the geometric insights, and then evaluate the proposed AO-based design in general multiuser systems.

\subsection{Validation of Geometric Insights in Section~\ref{sec:insight_pose_orientation} }
\begin{figure}[t]
  \centering
  \subfloat[Channel correlation]{
    \includegraphics[width=0.88\linewidth]{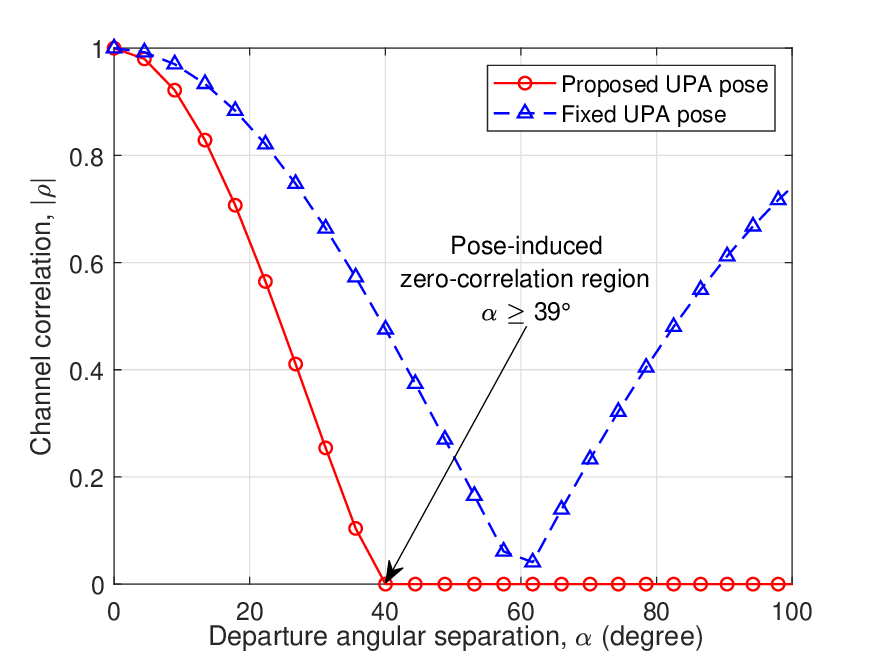}
    \label{fig:uniform1}
  } \\
  \subfloat[Weighted sum rate]{
    \includegraphics[width=0.88\linewidth]{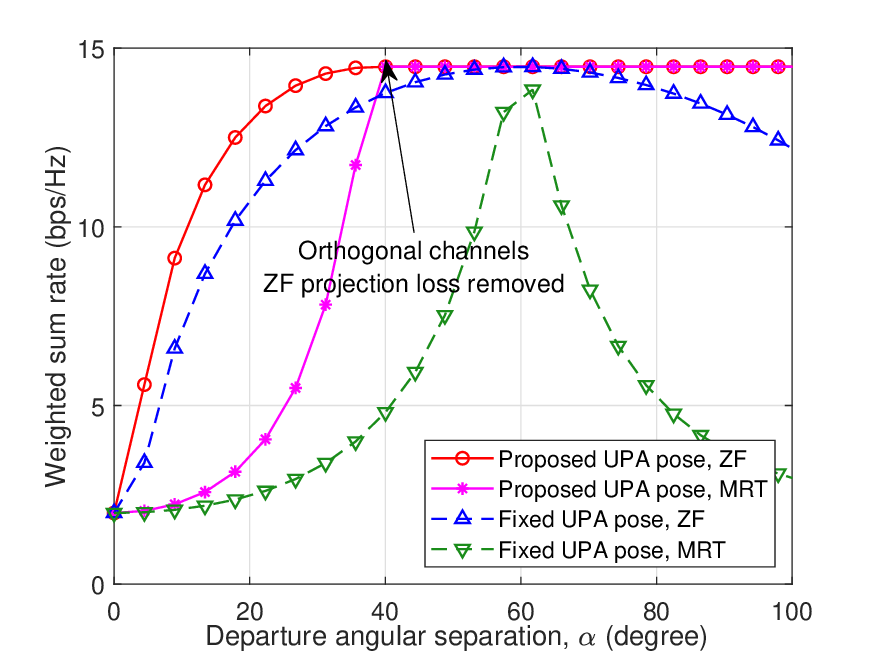}
    \label{fig:fibo1}
  }
  \caption{Pose-induced spatial separability under isotropic elements with $p=0$.}
  \label{fig:sim_two_user_isotropic}
\end{figure}

We first validate the geometric interpretation developed in Section~\ref{sec:insight_pose_orientation} under both isotropic and directional elements. Two users are symmetrically placed on a horizontal arc with radius $R_{\mathrm{arc}}=100$~m. By varying their arc separation, the horizontal axis represents the resulting three-dimensional departure angular separation $\alpha=\arccos(\bm d_1^T\bm d_2)$ observed from the BS. Unless otherwise specified, a $2\times3$ UPA and $P_{\max}=15$~dBm are used, and a LoS-only channel is adopted to isolate the far-field geometric mechanism. For comparison, the proposed curves follow the constructive solution in Section~\ref{sec:insight_pose_orientation}, whereas the fixed pose/boresight benchmark keeps the UPA normal pointing toward the positive $y$-direction with all element boresights aligned with the local array normal.

Fig.~\ref{fig:sim_two_user_isotropic} validates the isotropic-element insight in Section~\ref{subsec:p0_pose_orthogonalization}. When $p=0$, the element gain is angle-independent within the visible front half-space, and element-boresight steering provides no effective gain-shaping freedom. As shown in Fig.~\ref{fig:sim_two_user_isotropic}(a), for the considered $2\times3$ UPA with $d=\lambda/2$, the proposed UPA-pose design achieves a zero-correlation region once the angular separation exceeds about $39^\circ$. This agrees with Proposition~\ref{prop:pose_phase_orthogonalization}, since \eqref{eq:p0_pose_orth_condition_revised} gives $\alpha_{\rm th}=2\sin^{-1}(1/3)\approx38.9^\circ$. In contrast, the fixed UPA pose only exhibits an isolated correlation null around $60^\circ$, after which the correlation rises again due to its fixed array-factor pattern. This confirms that UPA-pose adaptation can project the user direction difference onto a favorable aperture dimension and achieve channel orthogonalization at a smaller angular separation. Fig.~\ref{fig:sim_two_user_isotropic}(b) shows the corresponding WSR under ZF and MRT beamforming. Consistent with the ZF decomposition in \eqref{eq:zf_wsr_decomposed_revised}, the proposed UPA-pose design improves the ZF rate in the small-to-moderate separation region by reducing channel correlation and hence the ZF projection loss. When $|\rho|$ approaches zero, the separability factor $1-|\rho|^2$ approaches one, and the ZF rate is no longer limited by inter-user channel alignment. The MRT curves further confirm this spatial-separability effect. When the UPA pose yields nearly orthogonal user channels, the MRT performance approaches that of ZF. In the limiting case of $\rho=0$, the normalized ZF directions coincide with the MRT directions and the two schemes achieve the same rates. Overall, under isotropic elements, UPA-pose adaptation turns the array response into a controllable spatial-separability DoF, enabling zero correlation beyond the first-null threshold and thereby removing the ZF projection loss.

\begin{figure}[t]
	\begin{center}
		\includegraphics[width=0.44\textwidth]{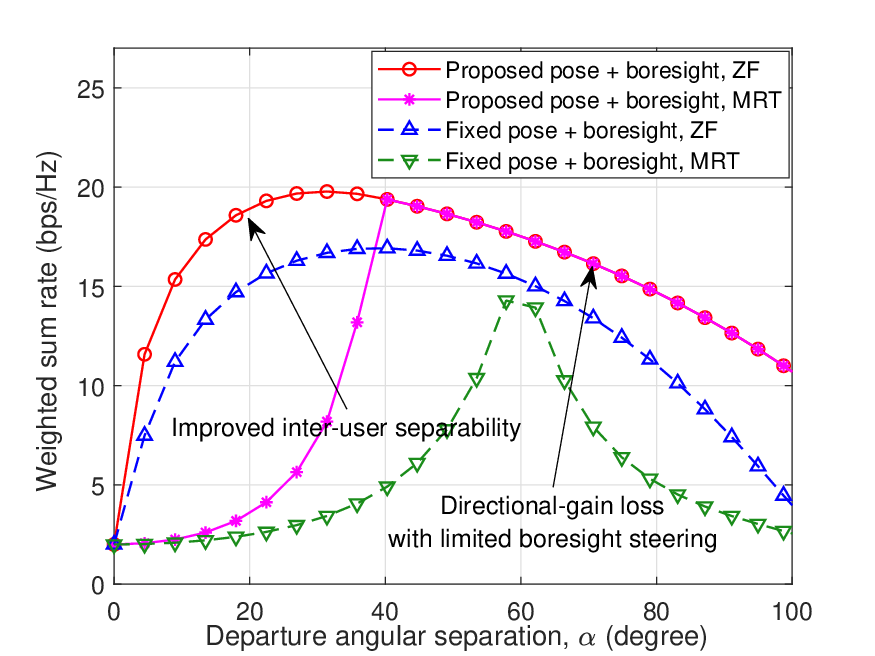}
		\caption{Pose--boresight coupling under directional elements with $p=4$.}
		\label{fig:twouserDir}
	\end{center}
\end{figure}

Fig.~\ref{fig:twouserDir} validates the pose--boresight coupling insight developed in Section~\ref{subsec:pose_orientation_roles} under directional elements with $p=4$ and $\theta_{\max}=\pi/4$, using the same topology as in Fig.~\ref{fig:sim_two_user_isotropic}. As $\alpha$ initially increases, the proposed pose-and-boresight design improves inter-user spatial separability by reshaping the array-induced phase progression while maintaining favorable directional-gain alignment, leading to a rapid WSR increase. Unlike the isotropic case, however, increasing the user separation is not always beneficial under directional elements. Beyond a moderate separation, the two users become more difficult to cover simultaneously due to the finite element-boresight steering range, and the WSR starts to decrease. This nonmonotonic trend reflects the key tradeoff between improving inter-user spatial separability and preserving sufficient directional gain. The MRT curves exhibit a similar nonmonotonic trend. For the proposed pose-and-boresight design, the MRT and ZF curves nearly merge when the user channels become sufficiently separable, because MRT suffers little residual inter-user interference and ZF incurs little projection loss. Overall, under directional elements, the UPA pose mainly shapes the array-induced phase progression and coarse directional coverage, whereas the element boresights refine the direction-dependent gains, motivating the joint design of UPA pose and element boresights.

\subsection{Validation of Proposed Algorithm in Section~\ref{sec:AO_solution}}
For the general multiuser simulations, unless otherwise stated, the BS is equipped with a $3\times3$ UPA whose initial normal points toward the positive $y$-direction in the GCS, and serves $K=6$ single-antenna users randomly dropped within the corresponding front-side semicircular service region of radius $R_{\mathrm{cell}}=100$ m. Each antenna element follows the directional pattern with the default factor $p=4$. We consider the following benchmarks in addition to the joint design:
\begin{itemize}
    \item \textbf{Pose only:} The UPA pose and beamformers are optimized, while all element boresights are fixed along the local UPA normal, as discussed in Remark~\ref{rem:ao_special_cases}.
    \item \textbf{Boresight only:} The element boresights and beamformers are optimized, while the UPA pose is fixed at its initial pose, as discussed in Remark~\ref{rem:ao_special_cases}.
    \item \textbf{Isotropic antenna:} The directional pattern is replaced by the front-side isotropic pattern by setting $p=0$, and only the beamformers are optimized under the fixed initial UPA pose and element boresights.
    \item \textbf{Fixed pose and boresights:} Both the UPA pose and element boresights are fixed at their initial values, and only the beamformers are optimized.
\end{itemize}

\begin{figure}[t]
	\begin{center}
		\includegraphics[width=0.44\textwidth]{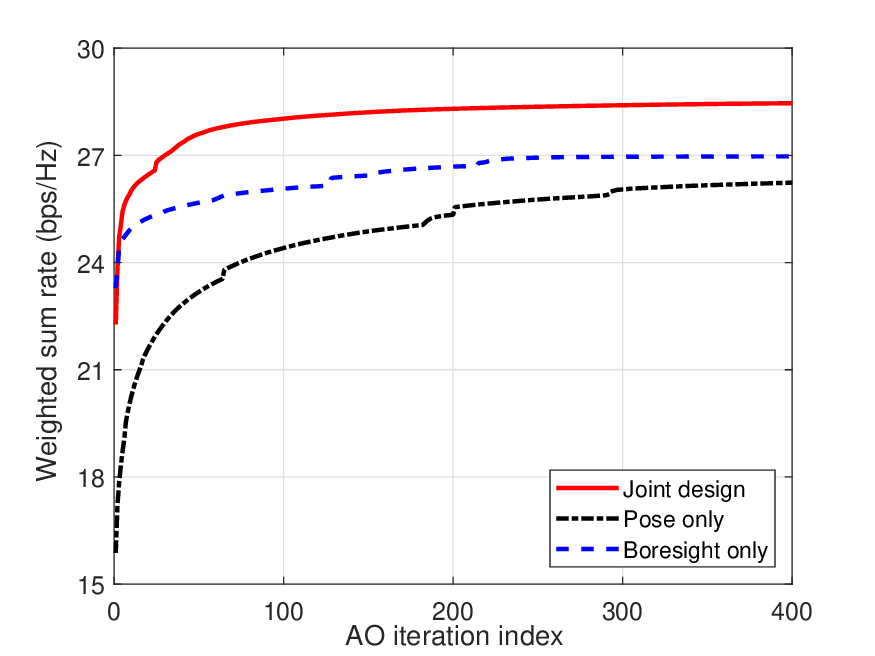}
		\caption{Convergence behavior of the AO-based algorithm.}
		\label{fig1}
	\end{center}
\end{figure}

Fig.~\ref{fig1} illustrates the convergence behavior of the three AO-based schemes, namely the joint design, pose-only design, and boresight-only design, where $P_{\max}=10$~dBm and $\theta_{\max}=\pi/4$. It is observed that all three schemes exhibit nondecreasing WSR over the AO iterations, which is consistent with the monotonicity of the block-wise updates established in Section \ref{sec:AO_solution}. The joint design reaches a stable WSR level within a few iterations, while the pose-only and boresight-only schemes also converge gradually as their respective geometric variables are updated together with the beamformers.

\begin{figure}[t]
	\begin{center}
		\includegraphics[width=0.44\textwidth]{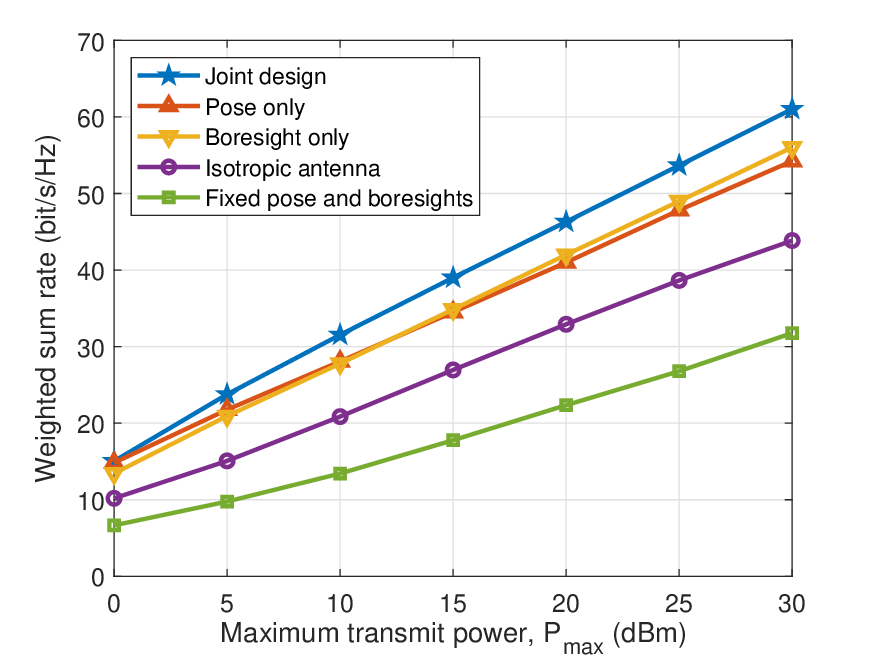}
		\caption{WSR vs. the transmit power under directional elements with $p=4$.}
		\label{fig:power_sweep}
	\end{center}
\end{figure}

\begin{figure}[t]
	\begin{center}
		\includegraphics[width=0.44\textwidth]{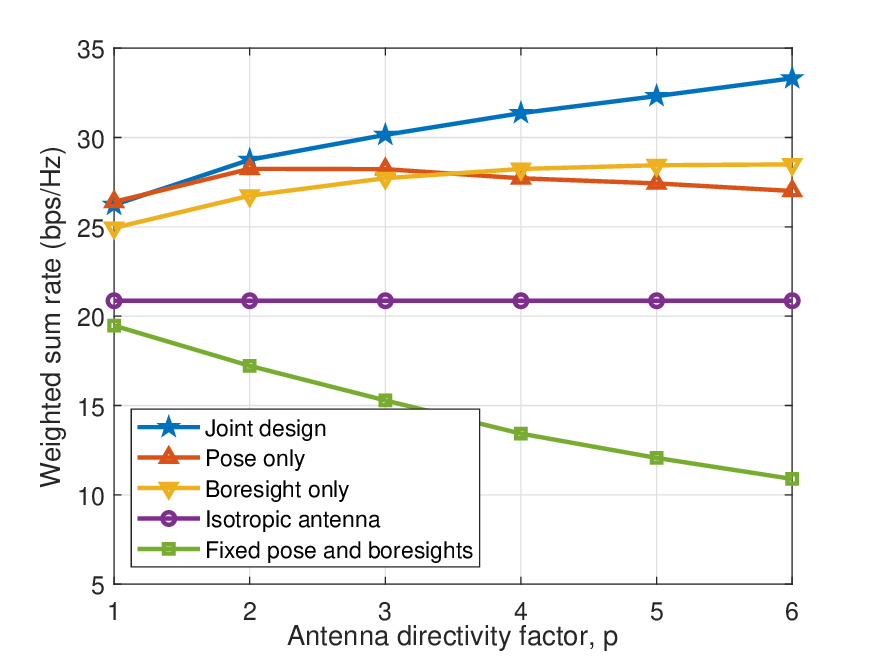}
		\caption{WSR vs. the antenna directivity factor.}
		\label{fig:p_sweep}
	\end{center}
\end{figure}

Fig.~\ref{fig:power_sweep} compares the WSR performance versus the maximum transmit power $P_{\max}$, where $\theta_{\max}=\pi/4$ and $p=4$. As expected, all schemes achieve higher WSR as $P_{\max}$ increases. The proposed joint design consistently outperforms all benchmarks over the entire power range, demonstrating the benefit of jointly exploiting UPA-wise pose adjustment and element-wise boresight steering. Both the pose-only and boresight-only designs provide substantial gains over the fixed-geometry baseline, confirming that each geometric DoF can improve the effective channels. Nevertheless, they remain inferior to the joint design because the UPA pose and element boresights are coupled: the pose determines both the array-induced phase difference and the local user directions available for boresight alignment, while the boresights further refine the direction-dependent gains. Notably, the fixed pose-and-boresight scheme performs even worse than the isotropic-element benchmark, especially at medium and high transmit powers. This shows that antenna directivity alone is not necessarily beneficial when the radiation directions are not properly matched to the user geometry. As $P_{\max}$ increases, the performance gap between the joint design and the benchmarks becomes more pronounced, since the system becomes less noise-limited and more sensitive to channel strength, spatial separability, and multiuser interference.

\begin{figure}[t]
	\begin{center}
		\includegraphics[width=0.44\textwidth]{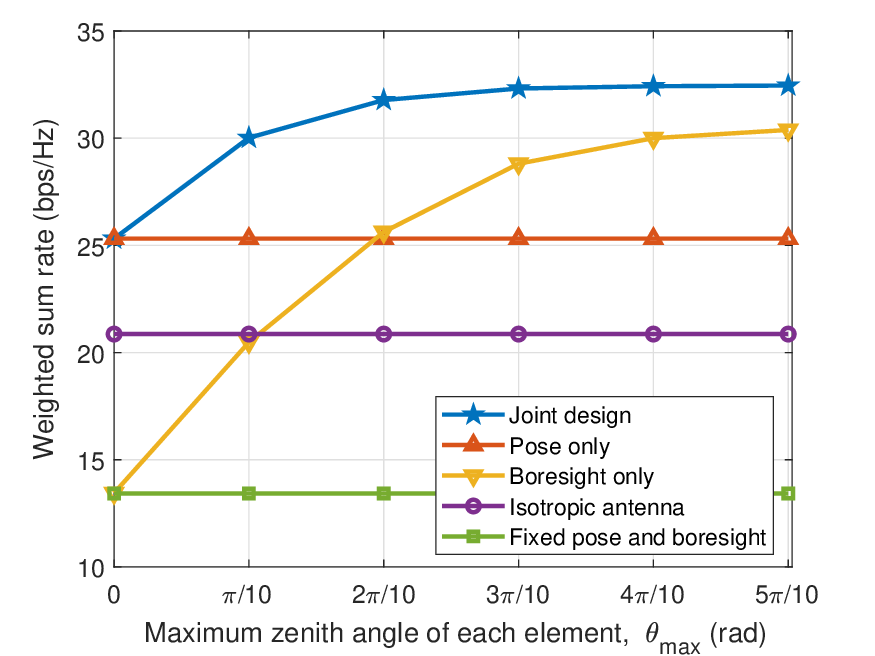}
		\caption{WSR vs. the maximum zenith angle.}
		\label{fig:theta_sweep}
	\end{center}
\end{figure}

\begin{figure}[t]
	\begin{center}
		\includegraphics[width=0.44\textwidth]{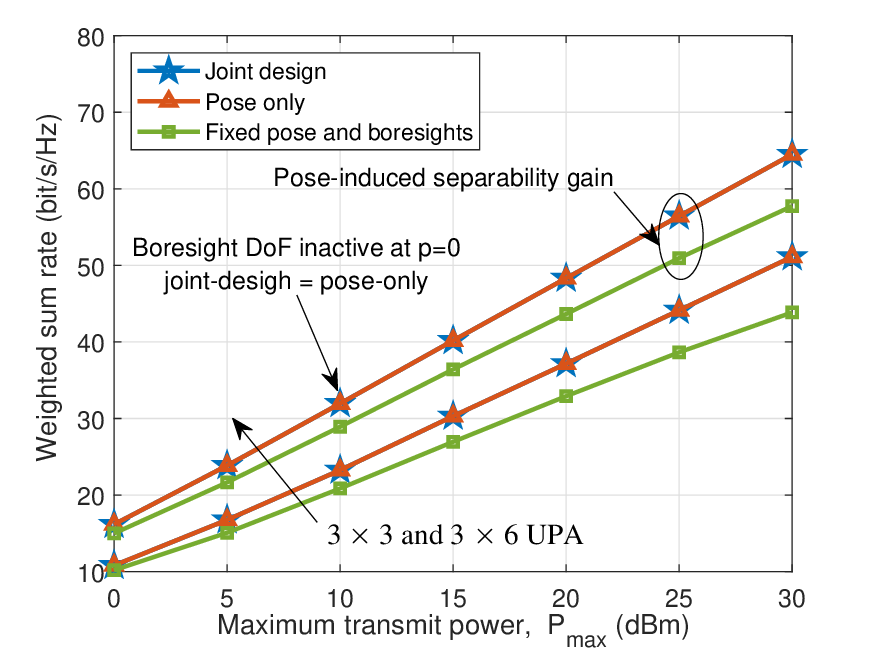}
		\caption{WSR vs. the transmit power under isotropic elements with $p=0$.}
		\label{fig:p0_pose_gain}
	\end{center}
\end{figure}

Fig.~\ref{fig:p_sweep} shows the WSR performance versus the element directivity factor $p$, where $P_{\max}=10$~dBm and $\theta_{\max}=\pi/4$. As $p$ increases, the element radiation pattern becomes narrower and more directive, providing a higher peak gain but also requiring more accurate directional alignment. The joint design achieves the highest WSR and benefits steadily from the increased directivity. The boresight-only design also gains from larger $p$, but its improvement is limited by the fixed UPA pose, which cannot reshape the array-induced phase geometry or the local user directions. In contrast, the pose-only design improves for small-to-moderate $p$ and then saturates or slightly degrades, because a single global UPA pose becomes insufficient to align narrow element beams toward multiple users. The fixed-pose-and-boresight benchmark degrades as $p$ increases, indicating that stronger directivity can hurt performance when the radiation directions are not properly matched to the user distribution. These results show that increased directivity is beneficial only when accompanied by proper pose and boresight adaptation.

Fig.~\ref{fig:theta_sweep} shows the WSR performance versus the maximum zenith steering angle $\theta_{\max}$, where $P_{\max}=10$~dBm and $p=4$. When $\theta_{\max}=0$, all element boresights are constrained to the local array normal, and hence the boresight-only scheme reduces to the fixed-pose-and-boresight benchmark, while the joint design reduces to the pose-only design. The pose-only, isotropic-elements, and fixed-pose-and-boresight schemes remain almost unchanged as $\theta_{\max}$ varies, since they do not exploit element-boresight steering. As $\theta_{\max}$ increases, the boresight-only scheme improves rapidly because a larger steering range enables more effective directional-gain alignment. However, for small steering ranges, the pose-only design can outperform the boresight-only design because UPA-pose optimization provides global array reorientation, coarse directional coverage, and array-phase shaping, whereas boresight-only steering remains limited by the fixed UPA pose. The joint design achieves the highest WSR by combining coarse pose-based spatial/directional adjustment with local boresight-based gain refinement.

Fig.~\ref{fig:p0_pose_gain} isolates the role of UPA-pose adaptation under isotropic elements by setting $p=0$, where the angle-dependent directional gain is removed. The maximum zenith angle is set to $\theta_{\max}=\pi/4$, and two UPA sizes, i.e., $3\times3$ and $3\times6$, are considered. Since the front-side element gain is independent of the boresight direction, element-boresight steering provides no gain-shaping DoF in this case; hence, the joint design and the pose-only design overlap. Nevertheless, both schemes outperform the fixed-pose-and-boresight benchmark, showing that UPA-pose optimization can improve the WSR even without element-level directivity. Consistent with the insight of Proposition~\ref{prop:pose_phase_orthogonalization}, this improvement mainly comes from reshaping the array-induced phase responses and enhancing inter-user spatial separability. Moreover, the $3\times6$ UPA achieves a higher WSR than the $3\times3$ UPA due to its larger aperture and more antenna elements, while the pose-induced separability gain remains visible for both array sizes.

\section{Conclusion}
This paper investigated multiuser downlink transmission enabled by a rotatable UPA with steerable element boresights. By jointly exploiting array-level pose control and element-level boresight steering, the proposed architecture reshapes both array-induced phase responses and direction-dependent channel gains. We formulated a WSR maximization problem over the transmit beamformers, UPA pose, and element boresights under practical visibility and steering constraints. To clarify the underlying mechanism, a two-user ZF analysis decomposed the rate into channel-strength and spatial-separability terms. For isotropic elements, boresight steering provides no effective gain-shaping freedom, while UPA rotation improves spatial separability by reducing inter-user channel correlation. For directional elements, a separability--gain tradeoff arises: the UPA pose controls array-induced phase progression and coarse directional coverage, whereas the element boresights refine direction-dependent gains. For the general multiuser transmission, we developed an AO-based algorithm that alternately updates the beamformers, UPA pose, and element boresights while preserving feasibility. Simulation results substantiated the proposed phase-and-gain channel-shaping principle. They showed that stronger element directivity is beneficial only when properly matched by pose and boresight adaptation, and that UPA-pose control is especially valuable under limited element steering. These results demonstrate the potential of pose--boresight reconfigurability as an effective geometry-domain approach for adaptive multiuser transmission.

\bibliographystyle{IEEEtran}

\bibliography{bibtex}

@ARTICLE{10555049,
  author={Wu, Qingqing and others},
  journal={Proc. IEEE}, 
  title={Intelligent surfaces empowered wireless network: Recent advances and the road to {6G}}, 
  year={2024},
  month={Jul.},
  volume={112},
  number={7},
  pages={724-763},
  doi={10.1109/JPROC.2024.3397910}}

@ARTICLE{10379539,
  author={Wang, Zhe and others},
  journal={IEEE Commun. Surveys Tut.}, 
  title={A tutorial on extremely large-scale {MIMO} for {6G}: Fundamentals, signal processing, and applications}, 
  year={2024},
  volume={26},
  number={3},
  pages={1560-1605},
  doi={10.1109/COMST.2023.3349276}}

@ARTICLE{10496996,
  author={Lu, Haiquan and others},
  journal={IEEE Commun. Surveys Tut.}, 
  title={A tutorial on near-field {XL-MIMO} communications toward {6G}}, 
  year={2024},
  volume={26},
  number={4},
  pages={2213-2257},
  doi={10.1109/COMST.2024.3387749}}

@ARTICLE{11053128,
  author={Galappaththige, Diluka and Mohammadi, Mohammadali and Ngo, Hien Quoc and Matthaiou, Michail and Tellambura, Chintha},
  journal={IEEE Trans. Commun.}, 
  title={Cell-Free Full-Duplex Communication—An Overview}, 
  year={2025},
  volume={73},
  number={11},
  pages={10468-10494},
  doi={10.1109/TCOMM.2025.3583637}}

@ARTICLE{8910627,
  author={Wu, Qingqing and Zhang, Rui},
  journal={IEEE Commun. Mag.}, 
  title={Towards smart and reconfigurable environment: Intelligent reflecting surface aided wireless network}, 
  year={2020},
  month={Jan.},
  volume={58},
  number={1},
  pages={106-112},
  doi={10.1109/MCOM.001.1900107}}

@ARTICLE{9998527,
  author={Zhang, Zijian and others},
  journal={IEEE Trans. Commun.}, 
  title={Active {RIS} vs. passive {RIS}: Which will prevail in {6G}?}, 
  year={2023},
  month={Mar.},
  volume={71},
  number={3},
  pages={1707-1725},
  doi={10.1109/TCOMM.2022.3231893}}

@ARTICLE{10795216,
  author={Munagala, Rakesh and others},
  journal={IEEE Trans. Commun.}, 
  title={{IRS}-aided multi-cell massive {MIMO} systems with impairments: {SE} analysis and optimization}, 
  year={2025},
  month={Jul.},
  volume={73},
  number={7},
  pages={5295-5312},
  doi={10.1109/TCOMM.2024.3516512}}

@ARTICLE{11007277,
  author={Wu, Qingqing and others},
  journal={IEEE Wireless Commun.}, 
  title={Intelligent Reflecting Surfaces for Wireless Networks: Deployment Architectures, Key Solutions, and Field Trials}, 
  year={2025},
  month={Dec.},
  volume={32},
  number={6},
  pages={141-148},
  doi={10.1109/MWC.001.250002}}

@ARTICLE{10753482,
  author={New, Wee Kiat and others},
  journal={IEEE Commun. Surveys Tut.}, 
  title={A Tutorial on Fluid Antenna System for {6G} Networks: Encompassing Communication Theory, Optimization Methods and Hardware Designs}, 
  year={2025},
  volume={27},
  number={4},
  pages={2325-2377},
  doi={10.1109/COMST.2024.3498855}}

@ARTICLE{10906511,
  author={Zhu, Lipeng and others},
  journal={IEEE Commun. Surveys Tut.}, 
  title={A Tutorial on Movable Antennas for Wireless Networks}, 
  year={2026},
  volume={28},
  number={},
  pages={3002-3054},
  doi={10.1109/COMST.2025.3546373}}

@ARTICLE{10286328,
  author={Zhu, Lipeng and Ma, Wenyan and Zhang, Rui},
  journal={IEEE Commun. Mag.}, 
  title={Movable antennas for wireless communication: Opportunities and challenges}, 
  year={2024},
  month={Jun.},
  volume={62},
  number={6},
  pages={114-120},
  doi={10.1109/MCOM.001.2300212}}

@ARTICLE{10243545,
  author={Ma, Wenyan and Zhu, Lipeng and Zhang, Rui},
  journal={IEEE Trans. Wireless Commun.}, 
  title={{MIMO} capacity characterization for movable antenna systems}, 
  year={2024},
  month={Apr.},
  volume={23},
  number={4},
  pages={3392-3407},
  doi={10.1109/TWC.2023.3307696}}

@ARTICLE{10354003,
  author={Zhu, Lipeng and Ma, Wenyan and Ning, Boyu and Zhang, Rui},
  journal={IEEE Trans. Wireless Commun.}, 
  title={Movable-Antenna Enhanced Multiuser Communication via Antenna Position Optimization}, 
  year={2024},
  volume={23},
  number={7},
  pages={7214-7229},
  doi={10.1109/TWC.2023.3338626}}

@ARTICLE{11049889,
  author={Tang, Boyi and others},
  journal={IEEE Trans. Wireless Commun.}, 
  title={Capacity Maximization of Uplink With Fluid Antenna System at Both Ends}, 
  year={2025},
  volume={24},
  number={12},
  pages={10580-10593},
  doi={10.1109/TWC.2025.3580506}}

@ARTICLE{10794752,
  author={Xu, Hao and others},
  journal={IEEE Trans. Commun.}, 
  title={Capacity Maximization for FAS-Assisted Multiple Access Channels}, 
  year={2025},
  volume={73},
  number={7},
  pages={4713-4731},
  doi={10.1109/TCOMM.2024.3516499}}

@ARTICLE{11039166,
  author={Ghadi, Farshad Rostami and others},
  journal={IEEE Trans. Wireless Commun.}, 
  title={Fluid Antenna Multiple Access With Simultaneous Non-Unique Decoding in Strong Interference Channel}, 
  year={2025},
  volume={24},
  number={12},
  pages={10183-10195},
  doi={10.1109/TWC.2025.3578280}}

@ARTICLE{11222668,
  author={Zheng, Beixiong and others},
  journal={IEEE Wireless Commun.}, 
  title={Rotatable antenna enabled wireless communication and sensing: Opportunities and challenges}, 
  year={2025},
  note={early access, doi: \url{10.1109/MWC.2025.3611919}},
  doi={10.1109/MWC.2025.3611919}}

@ARTICLE{11427014,
  author={Zheng, Beixiong and Wu, Qingjie and Ma, Tiantian and Zhang, Rui},
  journal={IEEE Trans. Commun.}, 
  title={Rotatable Antenna-Enabled Wireless Communication: Modeling and Optimization}, 
  year={2026},
  volume={74},
  number={},
  pages={6825-6842},
  doi={10.1109/TCOMM.2026.3672230}}

@ARTICLE{11489290,
  author={Peng, Xingxiang and Wu, Qingqing and Zheng, Ziyuan and Chen, Wen and Zhu, Yanze and Gao, Ying},
  journal={IEEE Trans. Wireless Commun.}, 
  title={Rotatable Antenna Enabled Spectrum Sharing: Joint Antenna Orientation and Beamforming Design}, 
  year={2026},
  month={Apr.},
  volume={25},
  number={},
  pages={15660-15674},
  doi={10.1109/TWC.2026.3683870}}

@ARTICLE{11520277,
  author={Zheng, Ziyuan and others},
  journal={IEEE J. Sel. Topics Sig. Process.}, 
  title={Low-Altitude {ISAC} with Rotatable Active and Passive Arrays}, 
  year={2026},
  note={early access, doi: \url{10.1109/JSTSP.2026.3693227}},
  doi={10.1109/JSTSP.2026.3693227}}

@ARTICLE{11534579,
  author={Zhang, Xuzhong and Xiang, Lin and Wang, Jiaheng and Gao, Xiqi},
  journal={IEEE Trans. Veh. Technol.}, 
  title={{UAV}-Aided Mmwave Massive {MIMO} Broadcasting With Rotatable Antenna Array}, 
  year={2026},
  note={early access, doi: \url{10.1109/TVT.2026.3696498}},
  doi={10.1109/TVT.2026.3696498}}

@ARTICLE{11206404,
  author={Zhang, Xuzhong and others},
  journal={IEEE Trans. Commun.}, 
  title={Rotatable Antenna Array Enabled {UAV} mmWave Massive {MIMO} Communication}, 
  year={2026},
  volume={74},
  number={},
  pages={1219-1236},
  doi={10.1109/TCOMM.2025.3622962}}

@ARTICLE{11134688,
  author={Xiong, Xue and others},
  journal={IEEE Wireless Commun. Lett.}, 
  title={Efficient channel estimation for rotatable antenna-enabled wireless communication}, 
  year={2025},
  month={Nov.},
  volume={14},
  number={11},
  pages={3719-3723},
  doi={10.1109/LWC.2025.3601979}}

@ARTICLE{11520842,
  author={Zhang, Guoying and others},
  journal={IEEE Wireless Commun. Lett.}, 
  title={Joint Antenna Rotation and {IRS} Beamforming for Multi-User Uplink Communications}, 
  year={2026},
  volume={15},
  number={},
  pages={3219-3223},
  doi={10.1109/LWC.2026.3693733}}

@ARTICLE{arxiv1,
  author={Peng, Xingxiang and others},
  journal={arXiv preprint: 2601.16543}, 
  title={Cell-Free {MIMO} with Rotatable Antennas: When Macro-Diversity Meets Antenna Directivity}, 
  year={2026},
  note={doi: \url{https://arxiv.org/abs/2601.16543}},
  number={}}

@book{Ant2016,
  title     = {Antenna theory: Analysis and design},
  author    = {C. A. Balanis},
  year      = {2016},
  publisher = {John Wiley \& Sons}}

@ARTICLE{5756489,
  author={Shi, Qingjiang and Razaviyayn, Meisam and Luo, Zhi-Quan and He, Chen},
  journal={IEEE Trans. Sig. Process.}, 
  title={An iteratively weighted {MMSE} approach to distributed sum-utility maximization for a {MIMO} interfering broadcast channel}, 
  year={2011},
  month={Sept.},
  volume={59},
  number={9},
  pages={4331-4340},
  doi={10.1109/TSP.2011.2147784}}

\end{document}